\documentclass[runningheads]{llncs2}

\usepackage[colorlinks=true]{hyperref}
\usepackage{enumitem}
\usepackage{amsmath}
\usepackage{amssymb}
\usepackage{xcolor}
\usepackage{pgfplots}
\pgfplotsset{compat=1.13} 
\usepackage{float}
\usepackage{color}
\usepackage{tikz}
\usetikzlibrary{arrows.meta}
\usetikzlibrary{arrows,automata}
\usetikzlibrary{shapes,decorations,patterns,positioning}
\usepackage{makecell}
\usepackage{subcaption}
\usepackage{tikz-3dplot}
\usepackage{listings}
\usepackage{multirow}

\def\techreport{}

\title{Equilibria-based Probabilistic Model Checking for Concurrent Stochastic Games}

\ifthenelse{\isundefined{\techreport}}{\author{Marta Kwiatkowska\inst{1}\orcidID{0000-0001-9022-7599} \and Gethin Norman\inst{2}\orcidID{0000-0001-9326-4344} \and David~Parker\inst{3}\orcidID{0000-0003-4137-8862} \and Gabriel Santos\inst{1}\orcidID{0000-0002-6570-9737}}
}{
\author{Marta Kwiatkowska\inst{1} \and Gethin Norman\inst{2} \and David~Parker\inst{3} \and Gabriel Santos\inst{1}}
}
\institute{Department of Computing Science, University of Oxford, UK
\and
School of Computing Science, University of Glasgow, UK
\and School of Computer Science, University of Birmingham, UK}

\titlerunning{Equilibria-based Probabilistic Model Checking for CSGs}
\authorrunning{Kwiatkowska \and Norman \and Parker \and Santos}

\usepackage{scalerel}

\newcommand\scale[2]{\vstretch{#1}{\hstretch{#1}{#2}}}

\newcommand{\appref}[1]{Appendix~\ref{#1}}
\newcommand{\sectref}[1]{Section~\ref{#1}}

\newcommand{\figref}[1]{Figure~\ref{#1}}
\newcommand{\tabref}[1]{Table~\ref{#1}}
\newcommand{\egref}[1]{Example~\ref{#1}}

\newcommand{\eqnref}[1]{(\ref{#1})}

\newcommand{\propref}[1]{Proposition~\ref{#1}}
\newcommand{\lemref}[1]{Lemma~\ref{#1}}
\newcommand{\defref}[1]{Definition~\ref{#1}}

\newcommand{\assumref}[1]{Assumption~\ref{#1}}

\newcommand{\defdefref}[2]{Definitions~\ref{#1} and \ref{#2}}

\newtheorem{assumption}{Assumption}
\newcounter{exampcount}
\newenvironment{examp}
{\refstepcounter{exampcount}
\vskip6pt\noindent
{\bf Example \arabic{exampcount}.}}
{}
\newcommand{\startpara}[1]{{%
\vskip6pt\noindent
{\bf #1.}}}

\def\ra{{\rightarrow}}

\def\Nset{\mathbb{N}}

\def\Rset{\mathbb{R}}

\def\Eset{\mathbb{E}}
\def\Rsetgeq{\mathbb{R}_{\geq 0}}

\def\ra{\rightarrow} %
\def\rmdef{\,\stackrel{\mbox{\rm {\tiny def}}}{=}}
\newcommand{\true}{\mathtt{true}} %

\renewcommand{\leq}{\leqslant}
\renewcommand{\geq}{\geqslant}

\def\squareforqed{\hbox{\rlap{$\sqcap$}$\sqcup$}}
\def\qed{\ifmmode\squareforqed\else{\unskip\nobreak\hfil
\penalty50\hskip1em\null\nobreak\hfil\squareforqed
\parfillskip=0pt\finalhyphendemerits=0\endgraf}\fi}

\newcommand\game{{\sf G}}
\newcommand\nfgame{{\sf N}}
\newcommand\mgame{{\sf Z}}

\newcommand\sinit{{\bar{s}}}

\newcommand\dist{{\mathit{Dist}}}

\newcommand\Prob{{\mathit{Prob}}}

\newcommand\val{{\mathit{val}}}

\newcommand{\last}{\mathit{last}}
\newcommand{\ipaths}{\mathit{IPaths}}
\newcommand{\fpaths}{\mathit{FPaths}}

\def\AP{{\mathit{AP}}}
\def\sat{{\,\models\,}}

\def\Sat{{\mathit{Sat}}}
\newcommand{\coalition}[1]{\langle \! \langle {#1} \rangle \! \rangle}
\def\notsat{{\,\not\models\,}}

\def\next{{X\,}}
\def\until{{\ {\cal U}\ }}
\def\buntil{{\ {\cal U}^{\leq k}\ }}

\def\next{{\mathtt X\,}}
\def\until{{\ \mathtt{U}\ }}
\def\untilop{{\mathtt{U}}}

\def\buntil{{\ \mathtt{U}^{\leq k}\ }}
\def\buntilop{{\mathtt{U}^{\leq k}}}

\newcommand{\buntilp}[1]{{\ \mathtt{U}^{\leq #1}\, }}

\def\future{{\mathtt{F}\ }}
\def\futureop{{\mathtt{F}\,}}

\def\bfuture{{\mathtt{F}^{\leq k}\ }}
\def\sbfuture{{\mathtt{F}^{\scale{.75}{\leq k}} }}
\def\bfutureop{{\mathtt{F}^{\leq k}\,}}

\newcommand\bfuturep[1]{{\mathtt{F}^{\leq #1}\ }}

\newcommand{\scumul}[1]{\mathtt{C}^{#1}} %

\newcommand{\ap}{\mathsf{a}}
\newcommand{\sinstant}[1]{\mathtt{I}^{#1}} %

\newcommand\V{{\mathtt V}}
\newcommand\probopP{{\mathtt P}}
\newcommand\nashop[3]{\coalition{#1}_{\max #2}(#3)}
\newcommand\snashop[3]{\coalition{#1}_{\scale{.75}{\max #2}}(#3)}

\newcommand\probop[2]{\probopP_{#1}[\,{#2}\,]}

\newcommand\rewopR{{\mathtt R}}
\newcommand\rewop[3]{\rewopR^{#1}_{#2}[\,{#3}\,]}

\newcommand{\lab}{{\mathit{L}}}

\newcommand{\rew}{\mathit{rew}}

\begin{document}

\maketitle

\begin{abstract}
Probabilistic model checking for stochastic games enables formal verification of systems that comprise competing or collaborating entities operating in a stochastic environment. Despite good progress in the area, existing approaches focus on zero-sum goals and cannot reason about scenarios where entities are endowed with different objectives. In this paper, we propose probabilistic model checking techniques for concurrent stochastic games based on Nash equilibria. We extend the temporal logic rPATL (probabilistic alternating-time temporal logic with rewards) to allow reasoning about players with distinct quantitative goals, which capture either the probability of an event occurring or a reward measure. We present algorithms to synthesise strategies that are subgame perfect social welfare optimal Nash equilibria, i.e., where there is no incentive for any players to unilaterally change their strategy in any state of the game, whilst the combined probabilities or rewards are maximised. We implement our techniques in the PRISM-games tool and apply them to several case studies, including network protocols and robot navigation, showing the benefits compared to existing approaches.
\end{abstract}

\section{Introduction}

Probabilistic model checking is a technique for formally verifying
systems that exhibit uncertainty or feature randomisation.
Quantitative system requirements, %
which express, e.g., safety, reliability or performance, 
are formally specified in temporal logic.
These are then automatically checked against a probabilistic model,
such as a Markov chain, %
capturing the possible behaviour of the system being verified.
Closely related is strategy synthesis,
which uses probabilistic models with nondeterminism, for example Markov decision processes (MDPs),
to automatically generate policies or controllers which
guarantee that pre-specified system requirements are satisfied.
Thanks to mature tool support~\cite{KNP11,DJKV17},
the methods have been successfully applied to many domains,
from autonomous vehicles, to computer security, to task scheduling.

Stochastic games are a modelling formalism that incorporates probability,
nondeterminism and multiple players, who can compete or collaborate to achieve their goals.
A variety of verification algorithms for these models have been devised, e.g.,~\cite{Cha07b,Umm10,AHK07,AM04,CAH13}.
More recently, probabilistic model checking and strategy synthesis techniques
for stochastic games have been proposed~\cite{CFK+13b,BKW17,KNPS18,KKKW18}
and implemented in the PRISM-games tool~\cite{KPW18}.
This has allowed modelling and verification of stochastic games
to be used for a variety of non-trivial applications,
in which competitive or collaborative behaviour between entities is a crucial ingredient, including computer security and energy management.

Initial work in this direction focused on \emph{turn-based} stochastic games (TSGs),
where each state is controlled by a single player~\cite{CFK+13b},
and proposed the logic rPATL, an extension of the well known logic ATL~\cite{AHK02}. The logic can specify that a coalition of players is able to achieve a quantitative objective regarding the probability of an event’s occurrence or the expectation of a reward measure, regardless of the strategies of the other players. 
Recently~\cite{KNPS18}, this was extended to \emph{concurrent} stochastic games (CSGs), in which players make decisions simultaneously.
This allows more realistic modelling of interactive agents operating %
concurrently.
In another direction, \emph{multi-objective} model checking of TSGs~\cite{CFKSW13,BKW17}
enabled reasoning about coalitions aiming to satisfy a \emph{Boolean
combination} of objectives, regardless of the remaining players' behaviour.

A limitation of these approaches 
is that they focus on \emph{zero-sum} properties, in which a coalition aims to satisfy some requirement or to optimise some objective,
while the remaining players have the directly opposing goal.
In this paper, we consider CSGs in which two coalitions of players
have distinct quantitative objectives.
For this, we use the notion of subgame perfect \emph{Nash equilibria}~\cite{OR04}, 
i.e., scenarios in which it is not beneficial for any player to unilaterally change their strategy in any state.
Furthermore, amongst these, we consider \emph{social welfare} optimal equilibria,
which maximise the sum of the objectives of the players.

We propose an extension to rPATL
which allows reasoning about subgame perfect social welfare optimal Nash equilibria between two coalitions of players, with respect to probabilistic or reward objectives,
expressed using a variety of temporal operators.
We then give a model checking algorithm for the logic against CSGs
which employs a combination of backwards induction (for finite-horizon operators)
and value iteration (for infinite-horizon operators).
A key ingredient of the computation
is finding social welfare optimal Nash equilibria for bimatrix games, which we perform using labelled polytopes~\cite{LH64} and a reduction to SMT.
We implement our techniques as an extension of the PRISM-games~\cite{KPW18} model checker and develop a selection of case studies, including robot navigation, communication protocols and power control, to evaluate its performance and applicability. We show that we are able to synthesise strategies that outperform those derived using existing techniques.
\ifthenelse{\isundefined{\techreport}}{%
}{%
}

\startpara{Related Work}
Game-theoretic models are used in many contexts within verification, as summarised above.
In addition, the existence of and the complexity of finding Nash equilibria for stochastic
games are studied in~\cite{CMJ04,Umm10},
but without practical algorithms. In \cite{PPB15}, a learning-based algorithm for finding Nash equilibria for discounted properties of CSGs is presented and evaluated. Similarly,~\cite{LP17} studies Nash equilibria for discounted properties and introduces iterative algorithms for strategy synthesis. A theoretical framework for price-taking equilibria of CSGs is given in~\cite{AY17}, where players try to minimise their costs which include a price common to all players and dependent on the decisions of all players.
A notion of strong Nash equilibria for a restricted class of CSGs is formalised in~\cite{DDJ+18} and an approximation algorithm for checking the existence of such equilibria for discounted properties is introduced and evaluated. We also mention~\cite{BMS16}, which studies the existence of stochastic equilibria with imprecise deviations for CSGs and proposes a PSPACE algorithm to compute such equilibria. 

For non-stochastic games, model checking tools such as
PRALINE~\cite{BRE13}, EAGLE~\cite{TGW15} and EVE~\cite{GNP+18} support Nash equilibria, as does MCMAS-SLK~\cite{CLM+14} via strategy logic.
General purpose tools such as Gambit~\cite{GAMB} can compute a variety
of equilibria but, again, not for stochastic games.

\section{Preliminaries}

We first provide some background material on game theory and stochastic games.
We let $\dist(X)$ denote the set of probability distributions over set $X$.
\begin{definition}[Normal form game] 
A (finite, $n$-person) \emph{normal form game} (NFG) is a tuple $\nfgame = (N,A,u)$ where: $N=\{1,\dots,n\}$ is a finite set of players; $A = A_1 {\times} \cdots {\times} A_n$ and $A_i$ is a finite set of actions available to player $i \in N$; $u {=} (u_1,\dots,u_n)$ and $u_i \colon A \rightarrow \Rset$ is a utility function for player $i \in N$.
\end{definition}
For an NFG $\nfgame$, the players choose actions at the same time, where the choice for player $i \in N$ is over the action set $A_i$. When each player $i$ choose $a_i$, the utility received by player $j$ equals $u_j(a_1,\dots,a_n)$. A (mixed) strategy $\sigma_i$ for player $i$ is a distribution over its action set.
A \emph{strategy profile} $\sigma{=} (\sigma_1,\dots,\sigma_n)$ is a tuple of strategies for each player and the expected utility of player $i$ under $\sigma$ is:
\[ \begin{array}{c}
u_i(\sigma) \ \rmdef \ \sum_{(a_1,\dots,a_n) \in A} u_i(a_1,\dots,a_n) \cdot \left( \prod_{j=1}^n \sigma_j(a_j) \right) \, .
\end{array} \]
For profile $\sigma{=}(\sigma_1,\dots,\sigma_n)$ and player $i$ strategy $\sigma_i'$, we define the sequence $\sigma_{-i} = (\sigma_1,\dots,\sigma_{i-1},\sigma_{i+1},\dots,\sigma_n)$ and profile $\sigma_{-i}[\sigma_i'] = (\sigma_1,\dots,\sigma_{i-1},\sigma_i',\sigma_{i+1},\dots,\sigma_n)$. For player $i$ and strategy sequence $\sigma_{-i}$, a \emph{best response} for player $i$ to $\sigma_{-i}$ is a strategy $\sigma^\star_i$ for player $i$ such that $u_i(\sigma_{-i}[\sigma^\star_i]) \geq u_i(\sigma_{-i}[\sigma_i])$ for all strategies $\sigma_i$ of player $i$.
We now introduce the concept of \emph{Nash equilibria} and a particular variant called \emph{social welfare optimal}, which are equilibria that maximise the total utility, i.e.\ maximise the sum of players' individual utilities.
\begin{definition}[Nash equilibrium]\label{nash-def}
For NFG $\nfgame$, a strategy profile $\sigma^\star$ is a \emph{Nash equilibrium} (NE) if $\sigma_i^\star$ is a best response to $\sigma_{-i}^\star$ for all $i \in N$. 
Furthermore $\sigma^\star$ is a \emph{social welfare optimal NE} (SWNE) if  $u_1(\sigma^\star){+}\cdots{+}u_n(\sigma^\star)\geq u_1(\sigma){+} \cdots{+}u_n(\sigma)$ for all Nash equilibria $\sigma$ of $\nfgame$.
\end{definition}
A two-player NFG is \emph{constant-sum} if there exists $c \in \Rset$ such that $u_1(\alpha) {+} u_2(\alpha) = c$ for all $\alpha \in A$ and \emph{zero-sum} if $c{=}0$. %
For general two-player NFGs, 
we have a \emph{bimatrix game} which can be represented by two distinct matrices $\mgame_1, \mgame_2 \in \Rset^{l \times m}$ where $A_1 {=} \{a_1,\dots,a_l\}$, $A_2 {=} \{b_1,\dots,b_m\}$, $z^1_{ij} = u_1(a_i,b_j)$ and $z^2_{ij} = u_2(a_i,b_j)$. 

\begin{examp}\label{1-eg} We consider a stag hunt game~\cite{PSS+11} where, if players decide to cooperate, this can yield a large payoff, but, if the others do not, then the cooperating player gets nothing while the remaining players get a small payoff. A scenario with 3 players, where two form a coalition, yields a bimatrix game: 
\vspace*{-0.2cm}
\[
\mgame_1 = 
\bordermatrix{
 & b_0 & b_1 & b_2 \cr 
 a_0 & 2 & 2 & 2  \cr 
 a_1 & 0 & 4 & 6 
 }
\qquad
\mgame_2 = 
\bordermatrix{
 & b_0 & b_1 & b_2  \cr 
 a_0 & 4 & 2 & 0  \cr 
 a_1 & 4 & 6 & 9
 }
\]
where $a_0$ and $a_1$ represent player 1 not cooperating and cooperating respectively and $b_i$ that $i$ players in the coalition cooperate.
There are three Nash equilibria:
\begin{itemize}
\item player 1 and the coalition select $a_0$ and $b_0$, respectively with utilities $(2,4)$;
\item player 1 selects $a_0$ and $a_1$ with probabilities $5/9$ and $4/9$ and the coalition selects $b_0$ and $b_2$ with probabilities $2/3$ and $1/3$ with utilities $(2,4)$;
\item player 1 and the coalition select $a_1$ and $b_2$ respectively with utilities $(6,9)$.
\end{itemize}
For instance, in the first case, neither player 1 nor the coalition thinks the other will cooperate: the best they can do is act alone. 
The third is the only SWNE.
\end{examp}

\startpara{Concurrent stochastic games} 
In this paper, we use CSGs,
in which players repeatedly make simultaneous (probabilistic) choices that update the game state.
\vspace*{-10pt}
\begin{definition}[Concurrent stochastic game] A \emph{concurrent stochastic multi-player game} (CSG) is a tuple
$\game = (N, S, \bar{S}, A, \Delta, \delta, \AP, \lab)$ where:
\begin{itemize}
\item $N=\{1,\dots,n\}$ is a finite set of players;
\item $S$ is a finite set of states and $\bar{S} \subseteq S$ is a set of initial states;
\item $A = (A_1\cup\{\bot\}) {\times} \cdots {\times} (A_n\cup\{\bot\})$ where $A_i$ is a finite set of actions available to player $i \in N$ and $\bot$ is an idle action disjoint from the set $\cup_{i=1}^n A_i$;
\item $\Delta \colon S \rightarrow 2^{\cup_{i=1}^n A_i}$ is an action assignment function;
\item $\delta \colon S {\times} A \rightarrow \dist(S)$ is a probabilistic transition function;
\item $\AP$ is a set of atomic propositions and $\lab \colon S \rightarrow 2^{\AP}$ is a labelling function.
\end{itemize}
\end{definition}
A CSG $\game$ starts in an initial state $\sinit \in \bar{S}$ and, when in state $s$, each player $i \in N$ selects an action from its available actions $A_i(s)$ given by $\Delta(s) \cap A_i$ if this set is non-empty and $\{ \bot \}$ otherwise. Supposing player $i$ selects action $a_i$, the state of the game is updated according to the distribution $\delta(s,(a_1,\dots,a_n))$. We augment CSGs with \emph{reward structures} of the form $r=(r_A,r_S)$ where $r_A \colon S{\times}A \ra \Rsetgeq$ is an action reward function and $r_S \colon S \ra \Rsetgeq$ is a state reward function.
\begin{definition}[End component] An end component of a CSG $\game$ is a pair $(S',\delta')$ comprising a subset $S' \subseteq S$ of states and a partial probabilistic transition function $\delta' \colon S' {\times} A \rightarrow \dist(S)$ satisfying the following conditions: 
\begin{itemize}
    \item $(S',\delta')$ defines a sub-CSG of $\game$, i.e., for all $s' \in S'$ and $\alpha \in A$, if $\delta'(s',\alpha)$ is defined, then $\delta'(s',\alpha) {=}\delta(s',\alpha)$ and $\delta'(s',\alpha)(s){=} 0$  for all $s \in S{\setminus}S'$;
    \item the underlying graph of $(S',\delta')$ is strongly connected. 
\end{itemize}
It is \emph{non-terminal} if $\delta(s,\alpha)(s') {>} 0$ for some $s \in S'$, $\alpha \in A$ and $s' \in S {\setminus} S'$.
\end{definition}
A path of $\game$ represents a resolution of both the players' and probabilistic choices and is given by a sequence $\pi = s_0 \xrightarrow{\alpha_0}s_1 \xrightarrow{\alpha_1} \cdots$ such that $s_i \in S$, $\alpha_i = (a^i_1,\dots,a^i_n) \in A$, $a^i_j \in A_j(s_i)$ for $j \in N$ and $\delta(s_i,\alpha_i)(s_{i+1})>0$ for all $i \geq 0$. For a path $\pi$, the $(i{+}1)$th state is denoted $\pi(i)$, the  $(i{+}1)$th action $\pi[i]$, and if $\pi$ is finite, the final state by $\last(\pi)$. %
The sets of finite and infinite paths (starting in state $s$) are given by $\fpaths_\game$ and $\ipaths_\game$ ($\fpaths_{\game,s}$ and $\ipaths_{\game,s}$).

\startpara{CSG strategies and equilibria}
A \emph{strategy} for player $i$ in a CSG $\game$ resolves the player's choices. More precisely, it is a function $\sigma_i \colon \fpaths_{\game} \ra \dist(A_i \cup \{ \bot \})$ such that if $\sigma_i(\pi)(a_i){>}0$, then $a_i \in A_i(\last(\pi))$. We denote by $\Sigma^i_\game$ the set of strategies of player $i$. %

As for NFGs,
a \emph{strategy profile} for $\game$ is a tuple $\sigma {=} (\sigma_1,\dots,\sigma_{n})$ of strategies for all players and, for a player $i$ strategy $\sigma_i'$, we  define the sequence $\sigma_{-i}$ and profile $\sigma_{-i}[\sigma_i']$ in the same way.
For strategy profile $\sigma$ and state $s$, we let
$\ipaths^\sigma_{\game,s}$ denote the
infinite paths from $s$ under the choices of $\sigma$.
We can define a probability measure $\Prob^{\sigma}_{\game,s}$ over the infinite paths $\ipaths^{\sigma}_{\game,s}$~\cite{KSK76} and, for random variable $X \colon \ipaths_{\game} \rightarrow \Rsetgeq$, the expected value $\Eset^{\sigma}_{\game,s}(X)$ of $X$ in $s$ with respect to $\sigma$.

An \emph{objective} (or utility function) for player $i$ of $\game$
is a random variable $X_i \colon \ipaths_{\game} \rightarrow \Rsetgeq$.
This can encode, e.g., the probability or expected cumulative reward for reaching a target.
NE for CSGs can be defined as for NFGs.
Since our model checking algorithm is based on backwards induction~\cite{SW+01,NMK+44}, we restrict attention to \emph{sub-game perfect} NE~\cite{OR04}, which are NE in \emph{every state} of the CSG.
In addition, for infinite-horizon objectives, the existence of NE is an open problem~\cite{BMS14} so, for such objectives, we use $\varepsilon$-NE,
which exist for any $\varepsilon{>}0$. 
\begin{definition}[Subgame perfect $\varepsilon$-NE] For CSG $\game$ and $\varepsilon{>}0$, a strategy profile $\sigma^\star$ is a \emph{subgame perfect $\varepsilon$-Nash equilibrium}
for objectives $\langle X_i \rangle_{i \in N}$ if and only if\/ $\Eset^{\sigma^\star}_{\game,s}(X_i) \geq \sup_{\sigma_i \in \Sigma_i} \Eset^{\sigma^\star_{-i}[\sigma_i]}_{\game,s}(X_i) - \varepsilon$ for all $i \in N$ and $s \in S$.
\end{definition}
Social welfare optimal variants of these equilibria
(SWNEs and $\varepsilon$-SWNEs) are defined for CSGs as for NFGs above (see \defref{nash-def}).
\begin{figure}[t]
\centering
\begin{tikzpicture}[auto,every node/.style={scale=0.9},initial text={},->,>=latex,thick]
\node[state, inner sep=-1pt, minimum size=0pt] (s0) {$\begin{array}{c}1{,}0 \\ 1{,}0 \end{array}$};

\node[scale=0.02, left = 0.9 of s0] (nod1) {};
\draw[-] (s0) edge node [above]{$(t_1,t_2)$}(nod1);

\node[state, inner sep=-1pt, minimum size=0pt, above left=0.2 and 2.2 of s0] (s1) {$\begin{array}{c}0{,}1 \\ 0{,}1 \end{array}$};
\node[state, inner sep=-1pt, minimum size=0pt, below left=0.2 and 2.2 of s0] (s2) {$\begin{array}{c}0{,}0 \\ 0{,}0 \end{array}$};
\draw[->] (nod1) edge node[left,xshift=-0mm,yshift=-1mm] {$\;\;\;q_2$} (s1);
\draw[->] (nod1) edge node[left,xshift=0mm,yshift=1mm] {$\;\;\;1{-}q_2$} (s2);

\node[state, inner sep=-1pt, minimum size=0pt, above right=0.2 and 1.75 of s0] (s3) {$\begin{array}{c}0{,}1 \\ 1{,}0 \end{array}$};

\node[state, inner sep=-1pt, minimum size=0pt, below right=0.2 and 1.75 of s0] (s4) {$\begin{array}{c}1{,}0 \\ 0{,}1 \end{array}$};

\node[state, inner sep=-1pt, minimum size=0pt, right=4
of s0] (s5) {$\begin{array}{c}0,1 \\ 0,1 \end{array}$};

\draw[->] (s0) edge node [above,xshift=-2mm]{$(t_1,w_2)$}(s3);
\draw[->] (s0) edge node [below,xshift=-2mm]{$(w_1,t_2)$}(s4);

\draw[->] (s4) edge node [below,xshift=2mm,yshift=-1mm]{$(t_1,w_2)$}(s5);
\draw[->] (s3) edge node [above,xshift=2mm,yshift=1mm]{$(w_1,t_2)$}(s5);

\draw[->] (s0) edge [loop above] node [above] {$(w_1,w_2)$} ();
\draw[->] (s3) edge [loop above] node [right,xshift=2mm,yshift=-2mm] {$(w_1,w_2)$} ();
\draw[->] (s4) edge [loop above] node [right,xshift=2mm,yshift=-2mm] {$(w_1,w_2)$} ();

\draw[->] (s1) edge [loop left] node [left] {$(w_1,w_2)$} ();
\draw[->] (s2) edge [loop left] node [left] {$(w_1,w_2)$} ();
\draw[->] (s5) edge [loop right] node [right] {$(w_1,w_2)$} ();
\end{tikzpicture}
\vspace*{-0.3cm}
\caption{CSG model of a medium access control problem.}\label{eg-fig}
\vspace*{-0.5cm}
\end{figure}
\begin{examp}\label{2-eg}
In \cite{BRE13} a deterministic concurrent game is used to model medium access control. Two users with limited energy share a wireless channel and choose to transmit ($t$) or wait ($w$) and, if both transmit, the transmissions fail due to interference. We extend this to a CSG by assuming that transmissions succeed with probability $q_2$ if both transmit. \figref{eg-fig} presents a CSG where each user has energy for one transmission (the first value of tuples labelling states represents if a user has energy and the second if it has successfully transmitted).

If the objectives are to maximise the probability of a successful transmission, there are two SWNEs when one user waits for the other to transmit and then transmits. This means both successfully transmit. If the objectives are to maximise the probability of being one of the first to transmit, then there is only one SWNE corresponding to both immediately trying to transmit.
\end{examp}

\section{Extending rPATL with Nash Formulae}\label{logic-sec}

We now extend the logic rPATL,
previously proposed for zero-sum properties of both TSGs~\cite{KNPS18} and CSGs~\cite{KNPS18},
to allow the analysis of equilibria-based properties.
Since we are limited to considering $\varepsilon$-SWNE for infinite-horizon properties, we assume some $\varepsilon$ has been fixed in advance when considering such properties.  
\begin{definition}[Extended rPATL syntax]
The syntax of our extended version of {\rm rPATL} is given by the grammar:
\begin{eqnarray*}
\phi & \; := \; & \mathtt{true} \; \mid \; \ap \; \mid \; \neg \phi \; \mid \; \phi \wedge \phi \; \mid \; \coalition{C}\probop{\sim q}{\psi} \; \mid \; \coalition{C}\rewop{r}{\sim x}{\rho} \; \mid \;   \nashop{C{:}C'}{\sim x}{\theta} \\
\theta & \; := \; & \probop{}{\psi}{+}\probop{}{\psi} \ \mid \  \rewop{r}{}{\rho}{+}\rewop{r}{}{\rho}  \\
\psi & \; := \; & \next \phi \ \mid \ \phi \buntil \phi \ \mid \ \phi \until \phi \\
\rho & \; := \; &  \sinstant{=k} \ \mid \ \scumul{\leq k} \ \mid \ \future \phi
\end{eqnarray*}
where $\ap$ is an atomic proposition, $C$ and $C'$ are coalitions of players such that $C' {=} N {\setminus} C$, $\sim \,\in \{<, \leq, \geq, >\}$, $q \in [0, 1]$, $x \in \Rset$, $r$ is a reward structure and $k \in \Nset$. %
\end{definition}
The logic rPATL is a branching-time temporal logic that combines the probabilistic operator $\probopP$ of PCTL~\cite{HJ94}, PRISM's reward operator $\rewopR$~\cite{KNP11}, and the coalition operator $\coalition{C}$ of ATL~\cite{AHK02}.
The formula $\coalition{C}\probop{\geq q}{\psi}$
states that the coalition $C$ has strategies which, when followed, regardless of the strategies of $N\backslash C$, guarantee that the probability of satisfying path formula $\psi$ is at least $q$.
Such properties are inherently \emph{zero-sum} in nature as one coalition tries to maximise an objective (here the probability of $\psi$) and the other to minimise it.

We extend rPATL with the ability to reason about \emph{equilibria} through \emph{Nash formulae} of the form $\nashop{C{:}C'}{\sim x}{\theta}$. In addition to the usual state ($\phi$), path ($\psi$) and reward ($\rho$) formulae,
we distinguish \emph{non-zero sum} formulae ($\theta$), which comprise a sum of probability or reward objectives.
The formula $\nashop{C{:}C'}{\sim x}{\theta}$ is satisfied
if there exists a subgame perfect SWNE strategy profile between coalitions $C$ and $C'({=}N {\setminus} C)$ under which the \emph{sum} of the two objectives in $\theta$
is ${\sim} x$.
As is common for probabilistic temporal logics, we allow numerical queries of the form $\coalition{C{:}C'}_{\max=?}[\theta]$ which return the
sum of SWNE values.

For probabilistic objectives ($\theta{=}\probop{}{\psi^1}{+}\probop{}{\psi^2}$),
each $\psi^i$ can be a ``next'' ($\mathtt{X}$), ``bounded until'' ($\buntilop$) or ``until'' ($\untilop$) operator, with the usual equivalences such as $\future\phi \equiv \true\until\phi$.
For reward objectives ($\theta{=}\rewop{r_1}{}{\rho^  1}{+}\rewop{r_2}{}{\rho^2}$),
each $\rho^i$ refers to the expected reward with respect to reward structure $r_i$:
the instantaneous reward after $k$ steps ($\sinstant{=k}$); the reward accumulated over $k$ steps ($\scumul{\leq k}$); or the reward accumulated until a state satisfying $\phi$ is reached ($\future \phi$).
\begin{examp}
Recall the medium access control CSG of \egref{2-eg}.
Formula $\nashop{p_1{:}p_2}{\geq 2}{\probop{}{\future \mathsf{send}_1}{+}\probop{}{\future \mathsf{send}_2}}$ means players $p_1$ and $p_2$ send their packets with probability 1, while
$\nashop{p_1{:}p_2}{=?}{\probop{}{\neg\mathsf{send}_2 \! \until \! \mathsf{send}_1}{+}\probop{}{\neg \mathsf{send}_1 \! \until \! \mathsf{send}_2}}$ asks what is the sum of subgame perfect SWNE values when the objectives are to maximise the probability of being one of the first to successfully transmit.
\end{examp}
\vskip0.5em
\noindent Before we give the semantics, we define \emph{coalition games} which, given a CSG and coalition (set of players), reduce the CSG to a two-player CSG. Without loss of generality we assume the coalition of players is of the form $C = \{1,\dots,n'\}$.
\begin{definition}[Coalition game] For CSG $\game{=}(N, S, \bar{s}, A, \Delta, \delta, \AP, \lab)$ and coalition $C{=} \{1,\dots,n'\} \subseteq N$, the \emph{coalition game} $\game^C {=} ( \{1,2\}, S, \bar{s}, A^C, \Delta, \delta^C, \AP, \lab )$ is a two-player game where: $A^C = A^C_1 {\times} A^C_2$, $A^C_1 =(A_1\cup\{\bot\}) {\times} \cdots {\times} (A_{n'}\cup\{\bot\})$, $A^C_2 = (A_{n'+1}\cup\{\bot\}) {\times} \cdots {\times} (A_n\cup\{\bot\})$ and
for any $s \in S$, $a^C_1=(a_1,\dots,a_{n'}) \in A^C_1$ and $a^C_2=(a_{n'+1},\dots,a_n) \in A^C_2$ we have $\delta^C(s,(a^C_1,a^C_2))=\delta(s,(a_1,\dots,a_n))$.
\vskip6pt \noindent
Furthermore, for a reward structure $r$ of $\game$, by abuse of notation we use $r$ for the corresponding reward structure of $\game^C$ which is constructed similarly.
\end{definition}
\begin{definition}[Extended rPATL semantics]\label{sem-def}
The satisfaction relation $\sat$ of our rPATL extension is defined inductively on the structure of the formula.
The propositional logic fragment $(\mathtt{true}$, $\ap$, $\neg$, $\wedge)$
is defined in the usual way.
For temporal operators and a state $s \in S$ in CSG $\game$, we have:
\begin{eqnarray*}
&\!\!\!s \sat \coalition{C} \probop{\sim q}{\psi}  
\; \Leftrightarrow \;
\exists \sigma_1 \in \Sigma^1 . \, \forall \sigma_2 \in \Sigma^2 . \,
\Prob^{\sigma_1,\sigma_2}_{\game^C,s} \big\{ \pi \in \ipaths^{\sigma_1,\sigma_2}_{\game^C,s} \mid \pi \sat \psi \big\} \sim q \\
&\!\!\!s \sat \coalition{C} \rewop{r}{\sim x}{\rho}
\;  \Leftrightarrow  \; 
\exists \sigma_1 \in \Sigma^1 . \, \forall \sigma_2 \in \Sigma^2  . \, \Eset^{\sigma_1,\sigma_2}_{\game^C,s}[\rew(r,\rho)] \sim x \\
&\!\!\!s \sat \nashop{C{:}C'}{\sim x}{\theta}
\; \Leftrightarrow  \; 
\exists \sigma_1^\star \in \Sigma^1, \sigma_2^\star \in \Sigma^2 . \, \left( \Eset^{\sigma_1^\star,\sigma_2^\star}_{\game^{C},s}(X^\theta_1)+\Eset^{\sigma_1^\star,\sigma_2^\star}_{\game^{C},s}(X^\theta_2) \right) \sim x
\end{eqnarray*}
and $(\sigma_1^\star,\sigma_2^\star)$ is a subgame perfect SWNE$\,$\footnote{In the case of infinite-horizon properties, this is a subgame perfect $\varepsilon$-SWNE.} for the objectives $(X^\theta_1,X^\theta_2)$ in $\game^{C}$ where,
for $1{\leq}i{\leq}2$ and $\pi \in \ipaths_{\game^C,s}^{\sigma_1,\sigma_2}:$
\begin{eqnarray*}
X^{\probop{}{\psi^1}{+}\probop{}{\psi^2}}_i(\pi) & \ = \ & 1 \;\mbox{if $\pi \sat \psi^i$ and 0 otherwise} \\
X^{\rewop{r_{\scale{.75}{1}}}{}{\rho^1}{+}\rewop{r_{\scale{.75}{2}}}{}{\rho^2}}_i(\pi) & \ = \ &  \rew(r_i,\rho^i)(\pi) \\
\pi \sat \next \phi 
& \ \Leftrightarrow \ & 
\pi(1) \sat \phi \\
\pi \sat \phi_1 \buntil \phi_2 
& \ \Leftrightarrow \ &
\exists i \leq k . \, (\pi(i) \sat \phi_2 \wedge \forall j < i . \, \pi(j) \sat \phi_1 )
\\
\pi \sat \phi_1 \until \phi_2 
& \ \Leftrightarrow \ & 
\exists i \in \Nset . \, ( \pi(i) \sat \phi_2 \wedge \forall j < i  . \, \pi(j) \sat \phi_1 ) \\
\rew(r,\sinstant{=k})(\pi) & = & r_S(\pi(k)) \\
\rew(r,\scumul{\leq k})(\pi) & = & \mbox{$\sum_{i=0}^{k-1}$} \big( r_A(\pi(i),\pi[i])+r_S(\pi(i)) \big) \\
\rew(r,\future  \phi)(\pi) & = & \left\{ \begin{array}{cl}
\infty
& \mbox{if} \; \forall j \in \Nset . \, \pi(j) \notsat \phi \\
\mbox{$\sum_{i=0}^{k_\phi}$} \big( r_A(\pi(i),\pi[i])+r_S(\pi(i)) \big) & \mbox{otherwise}
\end{array} \right.
\end{eqnarray*}
and $k_\phi = \min \{ k{-}1 \mid \pi(k) \sat \phi \}$.
\end{definition}

\section{Model Checking CSGs against Nash Formulae}\label{mc-sect}

Since rPATL is a branching-time logic, the basic model checking algorithm works by recursively computing the set $\Sat(\phi)$ of states satisfying formula $\phi$
over the structure of $\phi$. So, to extend the existing rPATL model checking algorithm for CSGs~\cite{KNPS18} to the logic from \sectref{logic-sec}, we need only consider Nash formulae $\nashop{C{:}C'}{\sim x}{\theta}$. %
This requires computation of subgame perfect SWNE values of the objectives $(X^\theta_1,X^\theta_2)$ and a comparison of their sum to the threshold $x$. %

We first explain how we compute SWNE values in bimatrix games, then subgame perfect SWNE values for finite-horizon objectives and lastly approximate subgame perfect $\varepsilon$-SWNE values for infinite-horizon objectives. We also discuss how to synthesise SWNE profiles. 
Our algorithm requires the following assumption on CSGs, which can be checked using standard graph-based methods. Without this assumption the presented value iteration algorithms are not guaranteed to converge 
\ifthenelse{\isundefined{\techreport}}{%
(for further details, see~\cite{extended}).
}{%
(see Appendices~\ref{2-app} and \ref{3-app} for further discussion of why this assumption is required).
}

\begin{assumption}\label{game-assum}
For any infinite-horizon probabilistic properties, there are no non-terminal end components.
For infinite-horizon reward properties, the targets are reached with probability 1 under all strategy profiles.
\end{assumption}

\startpara{Computing SWNE Values of Bimatrix Games} Finding Nash equilibria in bimatrix games is in the class of \emph{linear complementarity} problems (LCPs). More precisely, a profile $(\sigma_1,\sigma_2)$ is a Nash equilibrium of the bimatrix game $\mgame_1,\mgame_2 \in \Rset^{l \times m}$ where where $A_1 {=} \{a_1,\dots,a_l\}$, $A_2 {=} \{b_1,\dots,b_m\}$ if and only if there exists $u,v \in \Rset$ such that, for the column vectors $x \in \Rset^l$v and $y \in \Rset^m$ where $x_i {=} \sigma_1(a_i)$ and $y_j {=} \sigma_2(b_j)$ for $1 {\leq} i {\leq} l$ and $1 {\leq} j {\leq} m$, we have:
\[
x^T(\mathbf{1} u - \mgame_1 y) = 0, \;\;\;
y^T(\mathbf{1} v - \mgame_2^T x) = 0, \;\;\;
\mathbf{1} u -\mgame_1 y \geq \mathbf{0} , \;\;\;
\mathbf{1} v-\mgame_2^T x \geq \mathbf{0}
\]
and $\mathbf{0}$ and $\mathbf{1}$ are vectors or matrices with all components 0 and 1, respectively.

The Lemke-Howson algorithm~\cite{LH64} can be applied for finding Nash equilibria and is based on the method of \emph{labelled polytopes} \cite{NRTV07}. Other well-known methods include those based on \emph{support enumeration}~\cite{PNS04} and \emph{regret minimisation} \cite{SGC05}. %

\startpara{SWNE via Labelled Polytopes} Given a bimatrix game $\mgame_1,\mgame_2 \in \Rset^{l \times m}$, we denote the sets of deterministic strategies of players 1 and 2 by $I {=} \{1,\dots,l\}$ and $M {=} \{1,\dots,m\}$ and define $J {=} \{l{+}1,\dots,l{+}m\}$ by mapping $j \in M$
to $l{+}j \in J$. A \emph{label} is then defined as element of $I \cup J$. The sets of strategies for players 1 and 2 can be represented by:
\[
X = \{x \in \Rset^l \mid \mathbf{1} x = 1 \; \wedge \; x \geq \mathbf{0}\} \quad \mbox{and} \quad
Y = \{y \in \Rset^m \mid \mathbf{1} y = 1 \; \wedge \; y \geq \mathbf{0}\} \, .
\]
The strategy set $Y$ is then divided into regions $Y(i)$ and $Y(j)$ (polytopes) for $i \in I$ and $j \in J$ such that $Y(i)$ contains strategies for which the deterministic strategy $i$ of player 1 is a best response and $Y(j)$ contain strategies which choose action $j$ with probability zero:
\[
Y(i) = \{y \in Y \mid \forall k \in I . \; \mgame_1(i,:)y \geq \mgame_1(k,:)y \} \;\; \mbox{and} \; \; Y(j) = \{ y = Y \mid y_{j-l} = 0\}
\]
where $\mgame_1(i,:)$ is the $i$th row vector of $\mgame_1$. A vector $y$ is then said to have label $k$ if $y \in Y(k)$, for $k \in I \cup J$.
The strategy set $X$ is divided analogously into regions $X(j)$ and $X(i)$ for $j \in J$ and $i\in I$ and a vector $x$ has label $k$ if $x \in X(k)$, for $k \in I \cup J$.
A pair of vectors $(x,y) \in X {\times} Y$ is \emph{completely labelled} if the union of the labels of $x$ and $y$ equals $I \cup J$.
The Nash equilibria of the game equal the vector pairs that are completely labelled~\cite{LH64,SHP74}. 

Once all completely labelled vector pairs have been computed, one can calculate the corresponding set of values through matrix-vector multiplication. The pairs that maximise the sum of values correspond to SWNE strategies. In case of multiple SWNEs, we choose the values that are maximal for the first player, unless both players can get equal payoff, in which case we choose these.
\startpara{Computing Values of Nash Formulae}
For a formula $\nashop{C{:}C'}{\sim x}{\theta}$, if the objectives of the non-zero sum formula $\theta$ are both finite-horizon, we can use \emph{backwards induction}~\cite{SW+01,NMK+44} to compute (precise) subgame perfect SWNE values. Below, we give the cases for bounded probabilistic reachability and bounded cumulative reward objectives;
\ifthenelse{\isundefined{\techreport}}{%
the remaining cases can be found in~\cite{extended}.
}{%
for the remaining cases see \appref{mc-app}.
}
If both of the objectives are infinite-horizon,
we use \emph{value iteration}~\cite{CH08} to approximate subgame perfect SWNE values. Since there is not necessarily a unique pair of such values, the convergence criterion is applied to the sum of the two values computed, which \emph{is} unique.
Below, we give details for
probabilistic and expected reachability objectives;
\ifthenelse{\isundefined{\techreport}}{%
the remaining cases can be found in~\cite{extended}.
}{%
for the remaining cases see \appref{mc-app}.
}
Finally, for cases where there is a combination of finite- and infinite-horizon objectives,
we convert to having both infinite-horizon by modifying the game and formula in a standard manner for probabilistic model checking;
\ifthenelse{\isundefined{\techreport}}{%
see~\cite{extended} for the construction.
}{%
see \appref{mc-app} for the construction.
} The two key aspects of the value iteration algorithm are using SWNE to ensure uniqueness and solving an MDP when the target of one player has been reached.

We use the notation $\V_{\game^C}(s,\theta)$ for SWNE values of the objectives $(X^\theta_1,X^\theta_2)$ in state $s$ of $\game^C$. We also use $\probopP^{\max}_{\game,s}(\psi)$ and $\rewopR^{\max}_{\game,s}(r,\rho)$ for the maximum probability of satisfying $\psi$ and maximum expected reward for the random variable $\rew(r,\rho)$, respectively, in state $s$ when all players collaborate. These can be computed through standard MDP model checking~\cite{BdA95,dA99}.
\startpara{Bounded Probabilistic Reachability} If $\theta {=} \probop{}{\bfuturep{k_1} \phi^1}{+}\probop{}{\bfuturep{k_2} \phi^2}$, then we compute values of the objectives for the formulae $\theta_{n+n_1,n+n_2}=\probop{}{\bfuturep{n+n_1} \phi^1}+\probop{}{\bfuturep{n+n_2} \phi^2}$ for $0{\leq} n {\leq}k$ recursively, where $k{=}\min\{k_1,k_2\}$, $n_1 {=} k_1{-}k$ and $n_2 = k_2{-}k$. For state $s$, if $n{=}0$:
\[
\V_{\game^C}(s,\theta_{n_1,n_2}) = \left\{ \begin{array}{cl}
(\eta_{\phi^1}(s),\eta_{\phi^2}(s)) & \;\; \mbox{if $n_1{=}n_2{=}0$} \\
(\eta_{\phi^1}(s),\probopP^{\max}_{\game,s}(\bfuturep{n_2} \phi^2) & \;\; \mbox{else if $n_1{=}0$} \\
(\probopP^{\max}_{\game,s}(\bfuturep{n_1} \phi^1),\eta_{\phi^2}(s)) & \;\; \mbox{otherwise}
\end{array}\right.
\]
and if $n{>}0$:
\[
\V_{\game^C}(s,\theta_{n+n_1,n+n_2}) = \left\{ \begin{array}{cl}
(1,1) & \;\; \mbox{if $s \in \Sat(\phi^1) \cap \Sat(\phi^2)$} \\
(1,\probopP^{\max}_{\game,s}(\bfuturep{n+n_2} \phi^2)) & \;\; \mbox{else if $s \in \Sat(\phi^1)$} \\
(\probopP^{\max}_{\game,s}(\bfuturep{n+n_1} \phi^1),1) & \;\; \mbox{else if $s \in \Sat(\phi^2)$} \\
\val(\mgame_1, \mgame_2) & \;\; \mbox{otherwise}
\end{array} \right. 
\]
where $\eta_{\phi^i}(s)$ equals $1$ if $s \in \Sat(\phi^i)$ and $0$ otherwise for $1 {\leq} i {\leq} 2$, and $\val(\mgame_1, \mgame_2)$ equals SWNE values of the bimatrix game $(\mgame_1,\mgame_2)\in \Rset^{l \times m}$:
\[ \begin{array}{c}
z^l_{i,j} = \sum_{s' \in S} \delta(s,(a_i,b_j))(s') \cdot v^{s',l}_{(n-1)+n_l}
\end{array} \]
$1{\leq}l{\leq}2$ and $(v^{s',1}_{(n-1)+n_1},v^{s',2}_{(n-1)+n_2}) = \V_{\game^C}(s',\theta_{(n-1)+n_1,(n-1)+n_2})$ for all $s' \in S$.

\startpara{Bounded Cumulative Rewards} If $\theta = \rewop{r_1}{}{\scumul{\leq k_1}}{+}\rewop{r_2}{}{\scumul{\leq k_2}}$, then we compute values of the objectives for the formulae $\theta_{n+n_1,n+n_2}=\rewop{r_1}{}{\scumul{\leq n+n_1}}+\rewop{r_2}{}{\scumul{\leq n+n_2}}$ for $0{\leq} n {\leq}k$ recursively,  where $k{=}\min\{k_1,k_2\}$, $n_1 {=} k_1{-}k$ and $n_2 = k_2{-}k$. For state $s$, if $n{=}0$:
\[
\V_{\game^C}(s,\theta_{n_1,n_2}) = 
\left\{ \begin{array}{cl}
(0,0) & \;\; \mbox{if $n_1{=}n_2{=}0$} \\
(0,\rewopR^{\max}_{\game,s}(r_2,\scumul{\leq n_2})) & \;\; \mbox{else if $n_1{=}0$} \\
(\rewopR^{\max}_{\game,s}(r_1,\scumul{\leq n_1}),0) & \;\; \mbox{otherwise}
\end{array} \right.
\]
and if $n{>}0$, then  $\V_{\game^C}(s,\theta_{n+n_1,n+n_2})$ equals SWNE values of the bimatrix game $(\mgame_1, \mgame_2) \in \Rset^{l \times m}$: 
\[ \begin{array}{c}
z^l_{i,j} = r^l_S(s) + r^l_A(s,(a_i,b_j)) + \sum_{s' \in S} \delta(s,(a_i,b_j))(s') \cdot v^{s',l}_{(n-1)+n_l} 
\end{array} \]
$1{\leq}l{\leq}2$ and $(v^{s',1}_{(n-1)+n_1},v^{s',2}_{(n-1)+n_2}) = \V_{\game^C}(s',\theta_{(n-1)+n_1,(n-1)+n_2})$ for all $s' \in S$.

\startpara{Probabilistic Reachability}
If $\theta = \probop{}{\future \phi^1}{+}\probop{}{\future \phi^2}$, values can be computed through value iteration as the limit $\V_{\game^C}(s,\theta) = \lim_{n \ra \infty} \V_{\game^C}(s,\theta,n)$ where:
\[
\V_{\game^C}(s,\theta,n) = \left\{ \begin{array}{cl}
(1,1) & \;\; \mbox{if $s \in \Sat(\phi^1) \cap \Sat(\phi^2)$} \\
(1,\probopP^{\max}_{\game,s}(\future \phi^2)) & \;\; \mbox{else if $s \in \Sat(\phi^1)$} \\
(\probopP^{\max}_{\game,s}(\future \phi^1),1) & \;\; \mbox{else if $s \in \Sat(\phi^2)$} \\
(0,0) & \;\; \mbox{else if $n{=}0$} \\
\val(\mgame_1, \mgame_2) & \;\; \mbox{otherwise}
\end{array} \right. 
\]
where $\val(\mgame_1, \mgame_2)$ equals SWNE values of the bimatrix game $(\mgame_1,\mgame_2)\in \Rset^{l \times m}$:
\[ \begin{array}{c}
z^l_{i,j} = \sum_{s' \in S} \delta(s,(a_i,b_j))(s') \cdot v^{s',l}_{n-1} 
\end{array} \]
$1{\leq}l{\leq}2$ and $(v^{s',1}_{n-1},v^{s',2}_{n-1}) = \V_{\game^C}(s',\theta,n{-}1)$ for all $s' \in S$. 
\startpara{Expected Reachability} If $\theta = \rewop{r_1}{}{\future \phi^1}{+}\rewop{r_2}{}{\future \phi^2}$, values can be computed through value iteration as the limit 
$\V_{\game^C}(s,\theta) = \lim_{n \ra \infty} \V_{\game^C}(s,\theta,n)$ where:
\[
\V_{\game^C}(s,\theta,n) = \left\{ \begin{array}{cl}
(0,0) & \;\; \mbox{if $s \in \Sat(\phi^1) \cap \Sat(\phi^2)$ or $n{=}0$} \\
(0,\rewopR^{\max}_{\game,s}(r_2,\future \phi^2)) & \;\; \mbox{else if $s \in \Sat(\phi^1)$} \\
(\rewopR^{\max}_{\game,s}(r_1,\future \phi^1),0) & \;\; \mbox{else if $s \in \Sat(\phi^2)$} \\
\val(\mgame_1, \mgame_2) & \;\; \mbox{otherwise}
\end{array} \right. 
\]
where $\val(\mgame_1, \mgame_2)$ equals SWNE values of the bimatrix game $(\mgame_1,\mgame_2)\in \Rset^{l \times m}$:
\[ \begin{array}{c}
z^l_{i,j} = r^l_S(s) + r^l_A(s,(a_i,b_j)) + \sum_{s' \in S} \delta(s,(a_i,b_j))(s') \cdot v^{s',l}_{n-1} 
\end{array} \]
$1{\leq}l{\leq}2$ and $(v^{s',1}_{n-1},v^{s',2}_{n-1}) = \V_{\game^C}(s',\theta,n{-}1)$ for all $s' \in S$. 
\startpara{Strategy Synthesis} 
In addition to property verification, it is usually beneficial to perform \emph{strategy synthesis}, that is, construct a witness of the satisfaction of a property. In the case of a  formula $\nashop{C{:}C'}{\sim x}{\theta}$, we can return a subgame perfect SWNE for the objectives $(X^\theta_1,X^\theta_2)$. This is achieved using the approach above, both keeping track of a SWNE for the bimatrix game solved in each state and, when computing optimal values for MDPs, also performing strategy synthesis~\cite{KP13} (a strategy of the MDP is equivalent to a strategy profile of the CSG). We can then combine these generated profiles to yield a subgame perfect SWNE. The synthesised strategies require randomisation and memory. Memory is needed since choices change after a path formulae becomes true or a target reached and is required for finite-horizon properties. For infinite-horizon properties, the use of value iteration means only approximate $\varepsilon$-NE profiles are synthesised. However, for the case studies in \sectref{case-sect}, we find that all synthesised profiles are NE.

\startpara{Correctness and Complexity}
The proof of correctness is given in
\ifthenelse{\isundefined{\techreport}}{%
the extended version of the paper~\cite{extended}
}{%
Appendix~\ref{appx-correctness}} and shows that the values computed during value iteration correspond to subgame perfect SWNE values of finite game trees, and the values of these game trees converge uniformly and are bounded from below and above by the finite approximations of $\game^C$ and actual values of $\game^C$, respectively. A limitation of our approach, as for standard value iteration~\cite{HM18,KKK+18}, is that convergence of the values  does not give guarantees on the precision. Complexity is linear in the size of the formula, while finding NE for reachability objectives is EXPTIME~\cite{Cha06}. Value iteration requires solving an LCP problem of size $|A|$ for each state at every iteration, with the number of iterations depending on the convergence criterion. \sectref{case-sect} reports on efficiency in practice.
\section{Implementation and Tool Support}

We have extended PRISM-games~\cite{KPW18} with support for modelling and verification of CSGs against equilibria-based properties, building upon the CSG extension of~\cite{KNPS18}.
The tool and files for the case studies of \sectref{case-sect} are available from~\cite{files}.

\startpara{Modelling}
CSGs are specified using an extension of the PRISM modelling language,
in which behaviour is defined using probabilistic guarded commands 
of the form ${[}a{]}\ g \rightarrow u$, where $a$ is an action label, $g$ is a guard (a predicate over states) and $u$ is a probabilistic state update. If it is enabled (i.e., $g$ is true), an $a$-labelled transition can probabilistically update the model's state.

This language is adapted to CSGs in~\cite{KNPS18} by assigning modules to players and, in any state, letting each player choose between enabled commands of the corresponding modules (if no command is enabled, the player idles). One requirement of~\cite{KNPS18} was that the updates of all player were independent of each other;
we extend the language to remove this requirement, by allowing commands to be labelled with lists of actions ${[}a_1,\dots,a_n{]}$, and thus represent behaviour dependent on other players' choices.
Rewards are extended similarly so that an individual player's rewards 
can depend on the choices taken by multiple players.

\startpara{Implementation}
We have implemented model construction of CSGs for the language described above, and the model checking and strategy synthesis algorithms of \sectref{mc-sect}, extending the PRISM-games implementation of rPATL verification~\cite{KNPS18}.
We build on PRISM's Java-based `explicit' engine which uses sparse matrices, and add an SMT-based implementation for solving bimatrix games using Z3~\cite{Z3}. The set of all Nash equilibria for a bimatrix game are found by progressively querying the SMT solver for new profiles until the model becomes unsatisfiable. Structuring the problem using labelled polytopes, which can be expressed through conjunctions, disjunctions and linear inequalities, avoids non-linear arithmetic. As an optimisation, we also search for and filter out \emph{dominated strategies} as a precomputation step to reduce the calls to the solver. 

\section{Case Studies and Experimental Results}\label{case-sect}

We now present case studies and results to demonstrate the applicability of our approach and implementation, as well as the benefits of using equilibria.

\startpara{Efficiency and Scalability} Before describing the case studies, we first discuss the performance of the implementation. In \tabref{tab:results}, we show experiments run on a 2.10 GHz Intel Xeon using 16GB RAM. The table includes model statistics (players, states and transitions) and the time to construct the CSG and verify it;
the latter is split between CSG verification (including solving the bimatrix games) and the instances of MDP verification.
Our tool can analyse models with over 2 million states and 20 million transitions; all are solved in under 2 hours and most are considerably quicker. However, for models where players have choices in almost all states, only models with up to tens of thousands of states can be verified within 2 hours.
The majority of the time is spent solving bimatrix games, and therefore it is the number of choices of each coalition, rather than the number of players, that affects performance. 

\begin{table}[t]
\centering
{\scriptsize
\begin{tabular}{|c|r||r|r|r|r||r|r|r|} \hline
\multicolumn{1}{|c|}{Case study \& property} & 
\multicolumn{1}{c||}{Param.} & 
\multicolumn{4}{c||}{CSG statistics} & \multicolumn{1}{c|}{Constr.} & \multicolumn{2}{c|}{Verif. time (s)} \\ 
\cline{3-6} \cline{8-9}
\multicolumn{1}{|c|}{[parameters]} &
\multicolumn{1}{c||}{values} & 
\multicolumn{1}{c|}{Players} & 
\multicolumn{1}{c|}{States} & 
\multicolumn{1}{c|}{Choices} & 
\multicolumn{1}{c||}{Trans.} & \multicolumn{1}{c|}{time(s)} & \multicolumn{1}{c|}{MDP} & \multicolumn{1}{c|}{CSG} \\ \hline \hline
\multirow{4}{*}{\shortstack[c]{\emph{Aloha} \\
$\probop{}{\futureop \mathit{sent}_1}{+}\probop{}{\futureop \mathit{sent}_{2,3}}$ \\
$\mbox{[$b_\mathit{max}$,$D$]}$}}
& 2,8 & 3 &    17,057 &    19,713 &     42,654 & 0.6 & 0.6  & 21.4 \\
& 3,8 & 3 &    89,114 &    97,326 &    264,172 & 2.2 & 2.1  & 32.8 \\
& 4,8 & 3 &   449,766 &   474,898 &  1,655,479 & 10.9 & 10.6  & 49.9 \\
& 5,8 & 3 & 2,308,349 & 2,385,537 & 10,362,711 & 97.7 & 90.0  & 121.7 \\
\hline\hline
\multirow{4}{*}{\shortstack[c]{\emph{Robot coordination} \\
$\probop{}{\bfutureop \!\! \mathit{goal}_1}{+}\probop{}{\bfutureop\!\! \mathit{goal}_2}$ \\
$\mbox{[$l$,$k$]}$}}
& 10,10 & 3 &   9,802 &    66,514 &    543,524 & 1.4 & 2.0 & 27.2  \\
& 15,15 & 3 &  50,177 &   375,549 &  3,175,539 & 5.0 & 19.8 & 131.8 \\
& 20,20 & 3 & 159,202 & 1,249,434 & 10,738,004 & 15.4 & 136.3 & 928.7 \\
& 25,25 & 3 & 389,377 & 3,142,669 & 27,267,419 & 48.3 & 548.8 & 4,837.0 \\
\hline\hline
\multirow{4}{*}{\shortstack[c]{\emph{Medium access control} \\ $\rewop{r_{\scale{.75}{1}}}{}{\scumul{\leq k}}{+}\rewop{r_{\scale{.75}{2}}}{}{\scumul{\leq k}}$ \\
$\mbox{[$e_\mathit{max}$,$\mathit{k}$]}$}}
& 10,20  & 2 &    441 &   1,600 &   2,759 & 0.1 & - & 17.2 \\
& 20,40  & 2 &  1,681 &   6,400 &  11,119 & 0.2 & - & 127.5 \\
& 40,80  & 2 &  6,561 &  25,600 &  44,639 & 0.7 & - & 991.7 \\
& 80,160 & 2 & 25,921 & 102,400 & 178,879 & 1.3 & - & 6,937.0 \\
\hline \hline
\multirow{4}{*}{\shortstack[c]{\emph{Power control} \\
$\rewop{r_{\scale{.75}{1}}}{}{\futureop e_1{=}0}{+}\rewop{r_{\scale{.75}{2}}}{}{\futureop e_2{=}0}$ \\
$\mbox{[$e_\mathit{max}$,$\mathit{k}$]}$}}
& 4,20 & 2 &  2,346 &   6,802 &  13,574 & 0.2 & 0.2 & 3.0 \\
& 4,40 & 2 & 10,746 &  30,700 &  60,854 & 0.4 & 1.0 & 12.7 \\
& 8,20 & 2 &  4,010 &  14,545 &  31,654 & 0.3 & 0.4 & 5.2 \\
& 8,40 & 2 & 32,812 & 119,694 & 260,924 & 1.2 & 3.9 & 64.8 \\
\hline
\end{tabular}}
\vspace*{0.1cm}
\caption{Statistics for a representative set of CSG verification instances.}\label{tab:results}
\vspace*{-0.9cm}
\end{table}

\startpara{Investigating the Benefits of Equilibria Properties}
In each case study, we compare our results with the corresponding zero-sum properties~\cite{KNPS18}. E.g., for $\nashop{C{:}C'}{=?}{\probop{}{\future \phi_1}{+}\probop{}{\future \phi_2}}$, we compute the value and an optimal strategy $\sigma_C$ for coalition $C$ of the formula $\coalition{C}\probop{\max=?}{\future \phi_1}$, and then find the value of an optimal strategy for the coalition $C'$ for $\probop{\min=?}{\future \phi_2}$ and $\probop{\max=?}{\future \phi_2}$ in the MDP induced by CSG when $C$ follows $\sigma_C$. 
The aim is to showcase the advantages of cooperation as, in many real-world applications, agents’ goals are not strictly opposed.
As will be seen, all the presented results demonstrate that by using equilibrium properties at least one of the players gains and in almost all cases neither player loses 
(in the one case study where this is not the case the gains far outweigh the losses). The individual SWNE values for players need not be unique and, for all case studies (except Aloha in which the players goals are not symmetric), the values can be swapped to give alternative SWNE values.

\startpara{Robot Coordination}
Our first case study models a scenario in which two robots move concurrently over a grid of size $l{\times}l$. The robots start in diagonally opposite corners and try to reach the corner from which the other starts. A robot can move either diagonally, horizontally or vertically towards its goal and when it moves there is a probability ($q$) that it instead moves in an adjacent direction. E.g., if a robot moves north east, then with probability $q/2$ it will move north or east. If the robots enter the same cell, they crash and are unable to move again.

\begin{figure}[t]
\centering
\begin{subfigure}{.49\textwidth}
\centering
\tiny{
\begin{tikzpicture}
\begin{axis}[
    title style={yshift=-2ex},
    title={$l{=}10$},
    ylabel={Sum of Probabilities},
    xlabel={$k$},
    xmin=9, xmax=14,
    xtick={9,10,11,12,13,14},
    ymin=0, ymax=2,
    ytick={0,0.25,0.5,0.75,1,1.25,1.5,1.75,2},
    legend pos=south east,
    legend style={fill=none},
    ymajorgrids=true,
    grid style=dashed,
    height=4.25cm,
    width=0.95\textwidth,
    legend entries={
                {$\snashop{p_{\scale{.75}{1}}{:}p_{\scale{.75}{2}}}{=?}{\mathtt{P}{+}\mathtt{P}}$},
                \textit{$\coalition{\mathit{p_{\scale{.75}{1}}}}{\mathtt P_{\scale{.75}{\max}}}+{\mathtt P_{\scale{.75}{\max}}}$},
             \textit{$\coalition{\mathit{p_{\scale{.75}{1}}}}{\mathtt P_{\scale{.75}{\max}}}+{\mathtt P_{\scale{.75}{\min}}}$}  
                }]
\addlegendimage{mark=square*,red,mark size=1.5pt}
\addlegendimage{mark=*,blue,mark size=1.5pt}
\addlegendimage{mark=triangle,orange,mark size=1.5pt}
]

\addplot[mark=square*,red,mark size=1.5pt] table [x=k, y=peq, col sep=comma]{rbts_crd_p_reach.csv};
\addplot[mark=*,blue,mark size=1.5pt] table [x=k, y=pmxmx, col sep=comma]{rbts_crd_p_reach.csv};
\addplot[mark=triangle,orange,mark size=1.5pt] table [x=k, y=pmxmn, col sep=comma]{rbts_crd_p_reach.csv};

\end{axis}
\end{tikzpicture}
}
\end{subfigure}
\begin{subfigure}{.49\textwidth}
\centering
\tiny{
\begin{tikzpicture}
\begin{axis}[
    title style={yshift=-2ex},
    title={player 1 (solid lines) and player 2 (dashed lines)},
    ylabel={Probability},
    xlabel={$k$},
    xmin=2, xmax=9,
    xtick={2,3,4,5,6,7,8,9,10},
    ymin=0, ymax=1,
    ytick={0,0.25,0.5,0.75,1},
    legend pos=south east,
    legend style={fill=none},
    ymajorgrids=true,
    grid style=dashed,
    height=4.25cm,
    width=0.95\textwidth,
    legend entries={
                $l{=}3$,
                $l{=}4$,
                $l{=}5$,
                $l{=}6$,      
                $l{=}7$                
                }]
\addlegendimage{mark=square*,red,mark size=1.5pt}
\addlegendimage{mark=o*,blue,mark size=1.5pt}
\addlegendimage{mark=triangle*,magenta,mark size=1.5pt}
\addlegendimage{mark=+,cyan,mark size=1.5pt}
\addlegendimage{mark=square,orange,mark size=1.5pt}
]

\addplot[mark=square*,red,mark size=1.5pt] table [x=k, y=p1_3, col sep=comma]{rbts_crd_p_reach2.csv};
\addplot[dashed,mark=square*,red,mark size=1.5pt] table [x=k, y=p2_3, col sep=comma]{rbts_crd_p_reach2.csv};
\addplot[mark=o,blue,mark size=1.5pt] table [x=k, y=p1_4, col sep=comma]{rbts_crd_p_reach2.csv};
\addplot[dashed,mark=o*,blue,mark size=1.5pt] table [x=k, y=p2_4, col sep=comma]{rbts_crd_p_reach2.csv};
\addplot[mark=triangle*,magenta,mark size=1.5pt] table [x=k, y=p2_5, col sep=comma]{rbts_crd_p_reach2.csv};
\addplot[dashed,mark=triangle*,magenta,mark size=1.5pt] table [x=k, y=p1_5, col sep=comma]{rbts_crd_p_reach2.csv};
\addplot[mark=+,cyan,mark size=1.5pt] table [x=k, y=p1_6, col sep=comma]{rbts_crd_p_reach2.csv};
\addplot[dashed,mark=square,cyan,mark size=1.5pt] table [x=k, y=p2_6, col sep=comma]{rbts_crd_p_reach2.csv};
\addplot[mark=o,orange,mark size=1.5pt] table [x=k, y=p1_7, col sep=comma]{rbts_crd_p_reach2.csv};
\addplot[dashed,mark=square,orange,mark size=1.5pt] table [x=k, y=p2_7, col sep=comma]{rbts_crd_p_reach2.csv};
\end{axis}
\end{tikzpicture}
}
\end{subfigure}
\vspace*{-0.4cm}
\caption{Robot coordination: $\nashop{p_1{:}p_2}{=?}{\probop{}{\bfuture \mathit{goal}_1}{+}\probop{}{\bfuture \mathit{goal}_2}}$ ($q{=}0.1$)}\label{robot-fig}
\vspace*{-0.5cm}
\end{figure}

We suppose the robots try to maximise the probability of reaching their individual goals %
eventually and within a given number of steps ($k$).
If there is no bound and $l{\geq}4$, the SWNE strategies allow each robot to reach its goal with probability 1 (as time is not an issue, they can collaborate to avoid crashing). For the bounded case, in \figref{robot-fig} we have plotted both the sum of the probabilities for a grid of size 10 (left) and the probabilities of the individual players for different grid sizes (right) as $k$ varies. When there is only one route to each goal within the bound (along the diagonal), i.e.\ when $k=l{-}1$, the SWNE strategies of both robots take this route. In odd grids, there is a high chance of crashing, but also a chance one will deviate and the other reaches their goal. Initially, as the bound $k$ increases, for odd grids the SWNE values for the players are not equal (see \figref{robot-fig} right). Here, it is better overall for one to follow the diagonal and the other to take a longer route, as if both took the diagonal route, the chance of crashing increases, decreasing the chance of reaching their goals. %

\startpara{Aloha} 
This case study concerns three users trying to send packets using the slotted ALOHA protocol. In a time slot, if a single user tries to send a packet, there is a probability ($q$) that the packet is sent; as more users try and send, then the probability of success decreases. If sending a packet fails, the number of slots a user waits before resending is set according to an exponential backoff scheme. More precisely, each user maintains a backoff counter which it increases each time there is a failure (up to $b_\mathit{max}$) and, if the counter equals $k$, randomly chooses the slots to wait from $\{0,1,\dots,2^k{-}1\}$.

\begin{figure}[t]
\centering
\vspace*{-0.2cm}
\begin{subfigure}{.49\textwidth}
\centering
\tiny{
\begin{tikzpicture}
\begin{axis}[
    ylabel={Sum of probabilities},
    xlabel={$D$},
    xmin=0, xmax=10,
    xtick={0,1,2,3,4,5,6,7,8,9,10},
    ymin=0, ymax=2,
    ytick={0,0.5,1,1.5,2.0,2.5},
    legend pos=north west,
    legend style={fill=none},
    ymajorgrids=true,
    grid style=dashed,
    height=4.25cm,
    width=0.95\textwidth,
    legend entries={
                {$\snashop{p_{\scale{.75}{1}}{:}p_{\scale{.75}{2}}}{=?}{\mathtt{P}{+}\mathtt{P}}$},
                \textit{$\coalition{\mathit{p_{\scale{.75}{1}}}}{\mathtt P_{\scale{.75}{\max}}}+{\mathtt P_{\scale{.75}{\max}}}$},
                \textit{$\coalition{\mathit{p_{\scale{.75}{1}}}}{\mathtt P_{\scale{.75}{\max}}}+{\mathtt P_{\scale{.75}{\min}}}$}    
                },
legend pos=south east]
\addlegendimage{mark=square*,red,mark size=1.5pt}
\addlegendimage{mark=*,blue,mark size=1.5pt}
\addlegendimage{mark=triangle,orange,mark size=1.5pt}
]

\addplot[mark=square*,red,mark size=1.5pt] table [x=tmax, y=peq, col sep=comma]{alh_bckff_3p_tmax.csv};
\addplot[mark=*,blue,mark size=1.5pt] table [x=tmax, y=pmxmx, col sep=comma]{alh_bckff_3p_tmax.csv};
\addplot[mark=triangle,orange,mark size=1.5pt] table [x=tmax, y=pmxmn, col sep=comma]{alh_bckff_3p_tmax.csv};


\end{axis}
\end{tikzpicture}
}
\end{subfigure}
\begin{subfigure}{.49\textwidth}
\centering
\tiny{
\begin{tikzpicture}
\begin{axis}[
    ylabel={Probability},
    xlabel={$q$},
    xmin=0, xmax=1,
    xtick={0,0.1,0.2,0.3,0.4,0.5,0.6,0.7,0.8,0.9,1.0},
    ymin=0, ymax=1,
    ytick={0,0.2,0.4,0.6,0.8,1.0},
    legend pos=north west,
    legend style={fill=none},
    ymajorgrids=true,
    grid style=dashed,
    height=4.25cm,
    width=0.95\textwidth,
    legend entries={
                 Player 1 (SWNE),
                Player 2 (SWNE),
                Player 1 (Mx),
                Player 2 (Mx)
                },
legend pos=north west]
\addlegendimage{mark=square*,red,mark size=1.5pt}
\addlegendimage{mark=*,blue,mark size=1.5pt}
\addlegendimage{mark=triangle*,magenta,mark size=1.5pt}
\addlegendimage{mark=+,cyan,mark size=1.5pt}
]

\addplot[mark=square*,red,mark size=1.5pt] table [x=p11, y=player1, col sep=comma]{alh_bckff_3p_p11_players.csv};
\addplot[mark=*,blue,mark size=1.5pt] table [x=p11, y=player2, col sep=comma]{alh_bckff_3p_p11_players.csv};;

\addplot[mark=triangle*,magenta,mark size=1.5pt] table [x=p11, y=p1, col sep=comma]{alh_bckff_3p_p11_players.csv};
\addplot[mark=+,cyan,mark size=1.5pt] table [x=p11, y=p2, col sep=comma]{alh_bckff_3p_p11_players.csv};

\end{axis}
\end{tikzpicture}
}
\end{subfigure}
\vspace*{-0.4cm}
\caption{Aloha: $\nashop{p_1{:}\{p_2,p_3\}}{=?}{\probop{}{\future (\mathit{sent}_1 {\wedge} t{\leq}D)}{+}\probop{}{\future (\mathit{sent}_2 {\wedge} \mathit{sent}_3 {\wedge} t{\leq}D)}}$}\label{aloha-fig}
\vspace*{-0.4cm}
\end{figure}

We suppose three users try to maximise the probability of sending packets before a deadline $D$, with users 2 and 3 forming a coalition. \figref{aloha-fig} presents total values as $D$ varies (left) and individual values as $q$ varies (right). Through synthesis, we find the collaboration is dependent on $D$ and $q$. Given more time there is more chance for the users to collaborate sending in different slots, 
while if $q$ is large it is unlikely users need to repeatedly send, so again can send in different slots. 
As the coalition has more messages to send, their probabilities are lower. However, even for two users, the probabilities are different, since, although it is advantageous to collaborate and only one user tries first, if transmission fails, then both users try to send as this is the best option for their individual goals.

\startpara{Medium Access Control}
Our third case study extends the CSG model from \egref{2-eg} by assuming the probability of a successful transmission when a single user tries to transmit equals $q_1$ and the energy of each user is bounded by $e_\mathit{max}$.

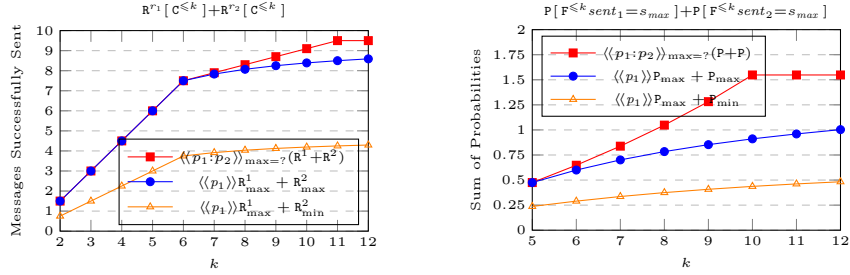
\begin{figure}[t]
\centering
\begin{subfigure}{.49\textwidth}
\centering
\tiny{
\begin{tikzpicture}
\begin{axis}[
    title style={yshift=-2ex},
    title={$\rewop{\scale{.75}{r_{\scale{.75}{1}}}}{}{\scumul{\scale{.75}{\leq k}}}{+}\rewop{\scale{.75}{r_{\scale{.75}{2}}}}{}{{\scumul{\scale{.75}{\leq k}}}}$},
    ylabel={Messages Successfully Sent},
    xlabel={$k$},
    xmin=2, xmax=12,
    xtick={2,3,4,5,6,7,8,9,10,11,12},
    ymin=0, ymax=10,
    ytick={0,1,2,3,4,5,6,7,8,9,10},
    legend pos=north west,
    legend style={fill=none},
    ymajorgrids=true,
    grid style=dashed,
    height=4.25cm,
    width=0.95\textwidth,
    legend entries={
                {$\snashop{p_{\scale{.75}{1}}{:}p_{\scale{.75}{2}}}{=?}{\mathtt{R^{\scale{.75}{1}}}{+}\mathtt{R^{\scale{.75}{2}}}}$},
                \textit{$\coalition{\mathit{p_{\scale{.75}{1}}}}{\mathtt R_{\scale{.75}{\max}}^{\scale{.75}{1}}}+{\mathtt R_{\scale{.75}{\max}}^{\scale{.75}{2}}}$},
                \textit{$\coalition{\mathit{p_{\scale{.75}{1}}}}{\mathtt R_{\scale{.75}{\max}}^{\scale{.75}{1}}}+{\mathtt R_{\scale{.75}{\min}}^{\scale{.75}{2}}}$}
                },
legend pos=south east]
\addlegendimage{mark=square*,red,mark size=1.5pt}
\addlegendimage{mark=*,blue,mark size=1.5pt}
\addlegendimage{mark=triangle,orange,mark size=1.5pt}
]
\addplot[mark=square*,red,mark size=1.5pt] table [x=k, y=eqrew, col sep=comma]{mdm_accss_r_bound.csv};
\addplot[mark=*,blue,mark size=1.5pt] table [x=k, y=mxmxrew, col sep=comma]{mdm_accss_r_bound.csv};
\addplot[mark=triangle,orange,mark size=1.5pt] table [x=k, y=mxmnrew, col sep=comma]{mdm_accss_r_bound.csv};

\end{axis}
\end{tikzpicture}
}
\end{subfigure}
\begin{subfigure}{.49\textwidth}
\centering
\tiny{
\begin{tikzpicture}
\begin{axis}[
    title style={yshift=-2ex},
    title={$\probop{}{\sbfuture \mathit{sent}_{\scale{.75}{1}}{=}s_{\scale{.75}{\mathit{max}}}}{+}\probop{}{\sbfuture \mathit{sent}_{\scale{.75}{2}}{=}s_{\scale{.75}{\mathit{max}}}}$},
    ylabel={Sum of Probabilities},
    xlabel={$k$},
    xmin=5, xmax=12,
    xtick={5,6,7,8,9,10,11,12},
    ymin=0, ymax=2,
    ytick={0,0.25,0.5,0.75,1,1.25,1.5,1.75,2},
    legend pos=north west,
    legend style={fill=none},
    ymajorgrids=true,
    grid style=dashed,
    height=4.25cm,
    width=0.95\textwidth,
    legend entries={
                {$\snashop{p_{\scale{.75}{1}}{:}p_{\scale{.75}{2}}}{=?}{\mathtt{P}{+}\mathtt{P}}$},
                \textit{$\coalition{\mathit{p_{\scale{.75}{1}}}}{\mathtt P_{\scale{.75}{\max}}}+{\mathtt P_{\scale{.75}{\max}}}$},
                \textit{$\coalition{\mathit{p_{\scale{.75}{1}}}}{\mathtt P_{\scale{.75}{\max}}}+{\mathtt P_{\scale{.75}{\min}}}$}    
                },
legend pos=north west]
\addlegendimage{mark=square*,red,mark size=1.5pt}
\addlegendimage{mark=*,blue,mark size=1.5pt}
\addlegendimage{mark=triangle,orange,mark size=1.5pt}
]

\addplot[mark=square*,red,mark size=1.5pt] table [x=k, y=peq, col sep=comma]{mdm_accss_cnt_p_reach.csv};
\addplot[mark=*,blue,mark size=1.5pt] table [x=k, y=pmxmx, col sep=comma]{mdm_accss_cnt_p_reach.csv};
\addplot[mark=triangle,orange,mark size=1.5pt] table [x=k, y=pmxmn, col sep=comma]{mdm_accss_cnt_p_reach.csv};

\end{axis}
\end{tikzpicture}
}
\end{subfigure}
\vspace*{-0.4cm}
\caption{Medium access control ($e_\mathit{max}{=}5$, $s_\mathit{max}{=}5$, $q_{1}{=}0.9$ and $q_{2}{=}0.75$)}\label{medium-fig}
\vspace*{-0.5cm}
\end{figure}

We consider two Nash properties for this model, both bounded by the number of time slots ($k$). The goal for each user in the first property is to maximise their expected number of successful transmissions and the second to maximise the probability of successfully transmitting a certain number ($s_\mathit{max}$) of messages. \figref{medium-fig} presents results for these properties as the bound $k$ varies. For both properties, the SWNE strategies yield equal values for the players. Synthesising strategies we see that for small values of $k$ there is not sufficient time to collaborate (both users always try and transmit); however, as $k$ increases there is time for the users to collaborate and try to transmit in different slots, and hence improve their values. Since the users have limited energy, \figref{medium-fig} shows that eventually adding steps does not increase the reward or probability.

\startpara{Power Control}
Our final case study is based a model of power control in cellular networks from~\cite{BRE13}. In the model, phones emit signals over a cellular network and the signals can be strengthened by increasing the power level up to a bound ($\mathit{pow}_\mathit{max}$). A stronger signal can improve transmission quality, but uses more energy and lowers the quality of other transmissions due to interference. We extend this model by adding a failure probability ($q_\mathit{fail}$) when a power level is increased and assume each phone has a limited battery capacity ($e_\mathit{max}$). Based on~\cite{BRE13}, we associate a reward structure with each phone representing transmission quality dependent both on its power level and that of other phones due to interference.

\begin{figure}[t]
\centering
\vspace*{-0.2cm}
\begin{subfigure}{.49\textwidth}
\centering
\tiny{
\begin{tikzpicture}
\begin{axis}[
    ylabel={Sum of Rewards},
    xlabel={$e_\mathit{max}$},
    xmin=1, xmax=5,
    xtick={1,2,3,4,5},
    ymin=3000, ymax=4000,
    ytick={3000,3200,3400,3600,3800,4000},
    legend pos=north west,
    legend style={fill=none},
    ymajorgrids=true,
    grid style=dashed,
    height=4.25cm,
    width=0.95\textwidth,
    legend entries={
                {$\snashop{p_{\scale{.75}{1}}{:}p_{\scale{.75}{2}}}{=?}{\mathtt{R^{\scale{.75}{1}}}{+}\mathtt{R^{\scale{.75}{2}}}}$},
                \textit{$\coalition{\mathit{p_1}}{\mathtt R_{\scale{.75}{\max}}^{\scale{.75}{1}}}+{\mathtt R_{\scale{.75}{\max}}^{\scale{.75}{2}}}$},
                \textit{$\coalition{\mathit{p_1}}{\mathtt R_{\scale{.75}{\max}}^{\scale{.75}{1}}}+{\mathtt R_{\scale{.75}{\min}}^{\scale{.75}{2}}}$}
                },
legend pos=north west]
\addlegendimage{mark=square*,red,mark size=1.5pt}
\addlegendimage{mark=*,blue,mark size=1.5pt}
\addlegendimage{mark=triangle,orange,mark size=1.5pt}
]

\addplot[mark=square*,red,mark size=1.5pt] table [x=pmax, y=eqrew, col sep=comma]{pwr_cntrl_r_reach.csv};
\addplot[mark=*,blue,mark size=1.5pt] table [x=pmax, y=mxmxrew, col sep=comma]{pwr_cntrl_r_reach.csv};;
\addplot[mark=triangle,orange,mark size=1.5pt] table [x=pmax, y=mxmnrew, col sep=comma]{pwr_cntrl_r_reach.csv};;

\end{axis}
\end{tikzpicture}
}
\end{subfigure}
\begin{subfigure}{.49\textwidth}
\centering
\tiny{
\begin{tikzpicture}
\begin{axis}[
    xlabel={$e_\mathit{max}$},
    xmin=1, xmax=5,
    xtick={1,2,3,4,5},
    ymin=1580, ymax=2030,
    ytick={1600,1700,1800,1900,2000},
    legend pos=north west,
    legend style={fill=none},
    ymajorgrids=true,
    grid style=dashed,
    height=4.25cm,
    width=0.95\textwidth,
    legend entries={
                Player 1 (SWNE),
                Player 2 (SWNE),
                Player 1 (Mx),
                Player 2 (Mx)
                },
legend pos=north west]
\addlegendimage{mark=square*,red,mark size=1.5pt}
\addlegendimage{mark=*,blue,mark size=1.5pt}
\addlegendimage{mark=triangle*,magenta,mark size=1.5pt}
\addlegendimage{mark=+,cyan,mark size=1.5pt}
]

\addplot[mark=square*,red,mark size=1.5pt] table [x=pmax, y=player1, col sep=comma]{pwr_cntrl_r_reach2.csv};
\addplot[mark=*,blue,mark size=1.5pt] table [x=pmax, y=player2, col sep=comma]{pwr_cntrl_r_reach2.csv};;

\addplot[mark=triangle*,magenta,mark size=1.5pt] table [x=pmax, y=mxmxrewp1, col sep=comma]{pwr_cntrl_r_reach3.csv};
\addplot[mark=+,cyan,mark size=1.5pt] table [x=pmax, y=mxmxrewp2, col sep=comma]{pwr_cntrl_r_reach3.csv};

\end{axis}
\end{tikzpicture}
}
\end{subfigure}
\vspace*{-0.4cm}
\caption{Power Control: $\nashop{p_1{:}p_2}{=?}{\rewop{r_1}{}{\future e_1{=}0}{+}\rewop{r_2}{}{\future e_2{=}0}}$}\label{power-fig}
\vspace*{-0.5cm}
\end{figure}

We consider two players, each trying to maximise their reward before their battery is empty. \figref{power-fig} presents, for $p_\mathit{max}{=}5$ and $e_\mathit{max}{=}5$, the sum of the SWNE values (left) and the values of the individual players (right) as the battery capacity varies. The values of the players are different because if one increases their power level this increases the overall reward (their reward increases, while the other's decreases by a lesser amount due to interference), whereas if both increase the overall reward decreases (both rewards decrease due to interference). %

\section{Conclusions}

We have presented a logic, algorithms and tool for model checking and strategy synthesis of concurrent stochastic games using Nash equilibria-based properties.
In comparison to existing methods, which support only zero-sum properties,
we demonstrate, on a range of case studies, that our approach produces
strategies that are collectively more beneficial for all players in the game.
Future work will investigate other techniques for Nash equilibria synthesis,
non-coalitional multi-player games and mechanism design.

\startpara{Acknowledgements} This work is partially supported by the EPSRC Programme Grant on Mobile Autonomy and the PRINCESS project, under the DARPA BRASS programme. We would like to thank the reviewers of an earlier version of this paper for finding a flaw in the correctness proof.

\bibliographystyle{spmpsci}
\bibliography{bib}

\ifthenelse{\isundefined{\techreport}}{%
}{%
\newpage
\appendix
\newpage
\section{Model Checking Procedure for Nash Formulae}\label{mc-app}

In this section we include the missing detail from \sectref{mc-sect} concerning computing values of Nash formulae. For the remainder of the section we fix a game $G$ and Nash formula $\nashop{C{:}C'}{\sim x}{\theta}$.

We first give the missing cases when the objectives are of the same type and then explain how to deal with a mixture of finite- and infinite-horizon objectives. We omit the cases when the objectives are both finite-horizon but of different types as these follow similarly to the cases when they are of the same type. We require the following notation used in \sectref{mc-sect}: for any formula $\phi$ and state $s$ let $\eta_{\phi}(s)$ equal $1$ if $s \in \Sat(\phi)$ and $0$ otherwise.
\startpara{Next} If $\theta = \probop{}{\next \phi^1}{+}\probop{}{\next \phi^2}$, then 
$\V_{\game^C}(s,\theta)$ equals SWNE values of the bimatrix game $(\mgame_1, \mgame_2) \in \Rset^{l \times m}$ where: 
\begin{align*}
z^1_{i,j} & = \; \mbox{$\sum_{s' \in S}$} \,  \delta(s,(a_i,b_j))(s') \cdot \eta_{\phi^1}(s') \\
z^2_{i,j} & = \; \mbox{$\sum_{s' \in S}$} \,  \delta(s,(a_i,b_j))(s') \cdot \eta_{\phi^2}(s') \, .
\end{align*}
\startpara{Bounded Until} If $\theta = \probop{}{\phi_1^1 \buntilp{k_1} \phi_2^1}\,+\,\probop{}{\phi_1^2 \buntilp{k_2} \phi_2^2}$, then we compute SWNE values for the objectives for the non-zero sum formulae $\theta_{n+n_1,n+n_2}=\probop{}{\phi_1^1 \buntilp{n+n_1} \phi_2^1}\,+\,\probop{}{\phi_1^2 \buntilp{n+n_2} \phi_2^2}$ for $0{\leq} n {\leq}k$ recursively,  where $k{=}\min\{k_1,k_2\}$, $n_1 {=} k_1{-}k$ and $n_2 {=} k_2{-}k$. For any state $s$, if $n{=}0$, then:
\[
\V_{\game^C}(s,\theta_{n_1,n_2}) = \left\{ \begin{array}{cl}
(\eta_{\phi^1_2}(s),\eta_{\phi^2_2}(s)) & \;\;\mbox{if $n_1{=}n_2{=}0$} \\
(\eta_{\phi^1_2}(s),\probopP^{\max}_{\game,s}(\phi_1^2 \buntilp{n_2} \phi_2^2)) & \;\;\mbox{else if $n_1{=}0$} \\
(\probopP^{\max}_{\game,s}(\phi_1^1 \buntilp{n_1} \phi_2^1),\eta_{\phi^2_2}(s)) & \;\;\mbox{otherwise.}
\end{array}\right.
\]
On the other hand, if $n{>}0$, then:
\[
\V_{\game^C}(s,\theta_{n+n_1,n+n_2}) = 
\left\{ \begin{array}{cl}
(1,1) & \;\; \mbox{if $s \in \Sat(\phi_2^1) \cap \Sat(\phi_2^2)$} \\
(1,\probopP^{\max}_{\game,s}(\phi_1^2 \buntilp{n+n_2} \phi_2^2)) & \;\; \mbox{else if $s \in \Sat(\phi_2^1)$} \\
(\probopP^{\max}_{\game,s}(\phi_1^1 \buntilp{n+n_1} \phi_2^1),1) & \;\; \mbox{else if $s \in \Sat(\phi_2^2)$} \\
(\probopP^{\max}_{\game,s}(\phi_1^1 \buntilp{n+n_1} \phi_2^1),0) & \;\; \mbox{else if $s \in \Sat(\phi_1^1) \setminus \Sat(\phi_1^2)$} \\
(0,\probopP^{\max}_{\game,s}(\phi_1^2 \buntilp{n+n_2} \phi_2^2)) & \;\; \mbox{else if $s \in \Sat(\phi_1^2) \setminus \Sat(\phi_1^1)$} \\
(0,0) & \;\; \mbox{else if $s \not\in \Sat(\phi_1^1) \cap \Sat(\phi_1^2)$} \\
\val(\mgame_1, \mgame_2) & \;\; \mbox{otherwise}
\end{array} \right. 
\]
where $\val(\mgame_1, \mgame_2)$ equals SWNE values of the bimatrix game $(\mgame_1,\mgame_2)\in \Rset^{l \times m}$:
\begin{align*}
z^1_{i,j} & = \; \mbox{$\sum_{s' \in S}$} \,   \delta(s,(a_i,b_j))(s') \cdot v^{s',1}_{(n-1)+n_1} \\ 
z^2_{i,j} & =  \;  \mbox{$\sum_{s' \in S}$} \,   \delta(s,(a_i,b_j))(s') \cdot v^{s',2}_{(n-1)+n_2}
\end{align*}
and $(v^{s',1}_{(n-1)+n_1},v^{s',2}_{(n-1)+n_2}) = \V_{\game^C}(s',\theta_{(n-1)+n_1,(n-1)+n_2})$ for all $s' \in S$.
\startpara{Bounded Instantaneous Rewards} If $\theta {=} \rewop{r_1}{}{\sinstant{=k_1}}{+}\rewop{r_2}{}{\sinstant{=k_2}}$, then we compute SWNE values of the objectives for the non-zero sum formulae $\theta_{n+n_1,n+n_2}=\rewop{r_1}{}{\sinstant{=n+n_1}}+\rewop{r_2}{}{\sinstant{=n+n_2}}$ for $0{\leq} n {\leq}k$ recursively,  where $k{=}\min\{k_1,k_2\}$, $n_1 = k_1{-}k$ and $n_2 {=} k_2{-}k$. For any state $s$, if $n{=}0$, then:
\[
\V_{\game^C}(s,\theta_{n_1,n_2}) =  
\left\{ \begin{array}{cl}
(r^1_S(s),r^2_S(s)) & \;\; \mbox{if $n_1{=}n_2{=}0$} \\
(r^1_S(s),\rewopR^{\max}_{\game,s}(r_2,\sinstant{=n_2})) & \;\; \mbox{else if $n_1{=}0$} \\
(\rewopR^{\max}_{\game,s}(r_1,\sinstant{= n_1}),r^2_S(s)) & \;\; \mbox{otherwise.}
\end{array} \right.
\]
On the other hand, if $n{>}0$, then
$\V_{\game^C}(s,\theta_{n+n_1,n+n_2})$ equals SWNE values of the bimatrix game $(\mgame_1, \mgame_2) \in \Rset^{l \times m}$ where: 
\begin{align*}
z^1_{i,j} & = \; \mbox{$\sum_{s' \in S}$} \,  \delta(s,(a_i,b_j))(s') \cdot v^{s',1}_{(n-1)+n_1} \\
z^2_{i,j} & = \; \mbox{$\sum_{s' \in S}$} \,   \delta(s,(a_i,b_j))(s') \cdot v^{s',2}_{(n-1)+n_2}
\end{align*}
and $(v^{s',1}_{(n-1)+n_1},v^{s',2}_{(n-1)+n_2}) = \V_{\game^C}(s',\theta_{(n-1)+n_1,(n-1)+n_2})$ for all $s' \in S$.
\startpara{Until}
If $\theta = \probop{}{\phi_1^1 \until \phi_2^1}{+}\probop{}{\phi_1^2 \until \phi_2^2}$, values for any state $s$ can be computed through value iteration as the limit $\V_{\game^C}(s,\theta) = \lim_{n \ra \infty} \V_{\game^C}(s,\theta,n)$ where:
\[
\V_{\game^C}(s,\theta,n) = \left\{ \begin{array}{cl}
(1,1) & \;\; \mbox{if $s \in \Sat(\phi_2^1) \cap \Sat(\phi_2^2)$} \\
(1,\probopP^{\max}_{\game,s}(\phi_1^2 \until \phi_2^2)) & \;\; \mbox{else if $s \in \Sat(\phi_2^1)$} \\
(\probopP^{\max}_{\game,s}(\phi_1^1 \until \phi_2^1),1) & \;\; \mbox{else if $s \in \Sat(\phi_2^2)$} \\
(\probopP^{\max}_{\game,s}(\phi_1^1 \until \phi_2^1),0) & \;\; \mbox{else if $s \in \Sat(\phi_1^1) \setminus \Sat(\phi_1^2)$} \\
(0,\probopP^{\max}_{\game,s}(\phi_1^2 \until \phi_2^2)) & \;\; \mbox{else if $s \in \Sat(\phi_1^2) \setminus \Sat(\phi_1^1)$} \\
(0,0) & \;\; \mbox{else if $n{=}0$ or $s \not\in \Sat(\phi_1^1) \cap \Sat(\phi_1^2)$} \\
\val(\mgame_1, \mgame_2) & \;\; \mbox{otherwise}
\end{array} \right. 
\]
where $\val(\mgame_1, \mgame_2)$ equals SWNE values of the bimatrix game $(\mgame_1,\mgame_2)\in \Rset^{l \times m}$:
\begin{align*}
z^1_{i,j} & = \; \mbox{$\sum_{s' \in S}$} \,  \delta(s,(a_i,b_j))(s') \cdot v^{s',1}_{n-1} \\
z^2_{i,j} & = \; \mbox{$\sum_{s' \in S}$} \, \delta(s,(a_i,b_j))(s') \cdot v^{s',2}_{n-1} 
\end{align*}
and $(v^{s',1}_{n-1},v^{s',2}_{n-1}) = \V_{\game^C}(s',\theta,n{-}1)$ for all $s' \in S$. \\ \\
We next demonstrate how to compute values of a Nash formula where the non-zero sum subformula contains both finite- and infinite-horizon objectives by reducing it to the computation of values of a Nash formula where the non-zero sum formula contains only infinite-horizon objectives on a modified game. We consider the case when the first objective is finite-horizon and second infinite-horizon, the symmetric cases follow similarly. In each case, the modified game has states of the form $(s,i)$ where $s \in S$ and $i \in \Nset$ and values for the state $(s,0)$ of the new Nash formula correspond to values of state $s$ for the original Nash formula.
\startpara{Next} For a non-zero sum formula of the form $\theta = \probop{}{\next \phi^1}{+}\probop{}{\phi_1^2 \until \phi_2^2}$, we construct the game $\game' = (N, S', \bar{S}', A, \Delta', \delta', \{ \ap_{\phi^1}, \ap_{\phi_1^2}, \ap_{\phi_2^2} \} , \lab')$ where:
\begin{itemize}
\item
$S' = \{ (s,i) \mid s \in S \wedge 0 {\leq} i {\leq} 1 \}$ and $\bar{S}' = \{ (s,0) \mid s \in S \}$;
\item $\Delta'((s,i)) = \Delta(s)$ for all $(s,i) \in S'$;
\item for any $(s,i),(s',i') \in S'$ and $a \in (A_1(s) \cup \{\bot\}){\times}(A_2(s) \cup \{\bot\})$:
\[
\delta'((s,i),a)((s',i')) = \left\{
\begin{array}{cl}
\delta(s,a)(s') & \;\; \mbox{if $i{=}0$ and $i' {=} i{+}1$ or $i{=}i'{=}1$} \\
0 & \;\; \mbox{otherwise;}
\end{array} \right.
\]
\item
for any $(s,i) \in S'$ and $1{\leq}j{\leq}2$:
\begin{itemize}
\item
$\ap_{\phi^1} \in \lab'((s,i))$ if and only if $s \in \Sat(\phi^1)$ and $i{=}1$;
\item
$\ap_{\phi_j^2} \in \lab'((s,i))$ if and only if $s \in \Sat(\phi_j^2)$.
\end{itemize}
\end{itemize}
Then compute values of $\nashop{C{:}C'}{\sim x}{\probop{}{\future \ap_{\phi^1}}{+}\probop{}{\ap_{\phi_1^2} \until \ap_{\phi_2^2}}}$ for the game $\game'$.
\startpara{Bounded Until} For a non-zero sum formula of the form $\theta = \probop{}{\phi_1^1 \buntilp{k_1} \phi_2^1}+\probop{}{\phi_1^2 \until \phi_2^2}$, we construct the game
\[
\game' = (N, S', \bar{S}', A, \Delta', \delta', \{ \ap_{\phi_1^1} , \ap_{\phi_2^1}, \ap_{\phi_1^2}, \ap_{\phi_2^2} \}, \lab')
\]
where:
\begin{itemize}
\item
$S' = \{ (s,i) \mid s \in S \wedge 0 {\leq} i {\leq} k_1{+}1 \}$ and $\bar{S}' = \{ (s,0) \mid s \in S \}$;
\item $\Delta'((s,i)) = \Delta(s)$ for all $(s,i) \in S'$;
\item for any $(s,i),(s',i') \in S'$ and $a \in (A_1(s) \cup \{\bot\}){\times}(A_2(s) \cup \{\bot\})$:
\[
\delta'((s,i),a)((s',i')) = \left\{
\begin{array}{cl}
\delta(s,a)(s') & \;\; \mbox{if $0 {\leq} i {\leq} k_1$ and $i' {=} i{+}1$ or $i{=}i'{=}k_1$} \\
0 & \;\; \mbox{otherwise;}
\end{array} \right.
\]
\item
for any $(s,i) \in S'$ and $1{\leq}j{\leq}2$:
\begin{itemize}
\item
$\ap_{\phi_1^1} \in \lab'((s,i))$ if and only if $s \in \Sat(\phi_1^1)$ and $0 {\leq} i {\leq} k_1{-}1$;
\item
$\ap_{\phi_2^1} \in \lab'((s,i))$ if and only if $s \in \Sat(\phi_2^1)$ and $0 {\leq} i {\leq} k_1$;
\item
$\ap_{\phi_j^2} \in \lab'((s,i))$ if and only if $s \in \Sat(\phi_j^2)$.
\end{itemize}
\end{itemize}
Then compute values of $\nashop{C{:}C'}{\sim x}{\probop{}{\ap_{\phi_1^1} \until \ap_{\phi_2^1}}{+}\probop{}{\ap_{\phi_1^2} \until \ap_{\phi_2^2}}}$ for the game $\game'$.
\startpara{Bounded Instantaneous Reward} For a non-zero sum formula of the form $\theta  = \rewop{r_1}{}{\sinstant{=k_1}}{+}\rewop{r_2}{}{\future \phi^2}$ we construct the game
\[
\game' = (N, S', \bar{S}', A, \Delta', \delta', \{ \ap_{k_1+1} , \ap_{\phi^2} \}, \lab')
\]
and reward structures $r_1'$ and $r_2'$ where:
\begin{itemize}
\item
$S' = \{ (s,i) \mid s \in S \wedge 0 {\leq} i {\leq} k_1{+}1 \}$ and $\bar{S}' = \{ (s,0) \mid s \in S \}$;
\item $\Delta'((s,i)) = \Delta(s)$ for all $(s,i) \in S'$;
\item for any $(s,i),(s',i') \in S'$ and $a \in (A_1(s) \cup \{\bot\}){\times}(A_2(s) \cup \{\bot\})$:
\[
\delta'((s,i),a)((s',i')) = \left\{
\begin{array}{cl}
\delta(s,a)(s') & \;\; \mbox{if $0 {\leq} i {\leq} k_1$ and $i' {=} i{+}1$ or $i{=}i'{=}k_1{+}1$} \\
0 & \;\; \mbox{otherwise;}
\end{array} \right.
\]
\item
for any $(s,i) \in S'$:
\begin{itemize}
\item
$\ap_{k_1+1} \in \lab'((s,i))$ if and only if $i {=} k_1{+}1$;
\item
$\ap_{\phi^2} \in \lab'((s,i))$ if and only if $s \in \Sat(\phi^2)$;
\end{itemize}
\item
for any $(s,i) \in S'$ and $a \in (A_1(s) \cup \{\bot\}){\times}(A_2(s) \cup \{\bot\})$:
\begin{itemize}
\item
$r^{1'}_A((s,i),a)=0$ and $r^{1'}_S((s,i))=r_S^1(s)$ if $i{=}k_1$ and equals 0 otherwise;
\item
$r^{2'}_A((s,i),a)=r^2_A(s)(a)$ and $r^{2'}_S((s,i))=r_S^2(s)$.
\end{itemize}
\end{itemize}
Then compute values of $\nashop{C{:}C'}{\sim x}{\rewop{r_1'}{}{\future \ap_{k_1+1}}{+}\rewop{r_2'}{}{\future \ap_{\phi^2}}}$ for the game $\game'$.
\startpara{Bounded Cumulative Reward} For a non-zero sum formula of the form $\theta = \rewop{r_1}{}{\scumul{\leq k_1}}{+}\rewop{r_2}{}{\future \phi^2}$ we construct the game \[
\game' = (N, S', \bar{S}', A, \Delta', \delta', \{ \ap_{k_1} , \ap_{\phi^2} \}, \lab')
\]
and reward structures $r_1'$ and $r_2'$ where:
\begin{itemize}
\item
$S' = \{ (s,i) \mid s \in S \wedge 0 {\leq} i {\leq} k_1 \}$ and $\bar{S}' = \{ (s,0) \mid s \in S \}$;
\item $\Delta'((s,i)) = \Delta(s)$ for all $(s,i) \in S'$;
\item for any $(s,i),(s',i') \in S'$ and $a \in (A_1(s) \cup \{\bot\}){\times}(A_2(s) \cup \{\bot\})$:
\[
\delta'((s,i),a)((s',i')) = \left\{
\begin{array}{cl}
\delta(s,a)(s') & \;\; \mbox{if $0 {\leq} i {\leq} k_1{-}1$ and $i' {=} i{+}1$ or $i{=}i'{=}k_1$} \\
0 & \;\; \mbox{otherwise;}
\end{array} \right.
\]
\item
for any $(s,i) \in S'$:
\begin{itemize}
\item
$\ap_{k_1} \in \lab'((s,i))$ if and only if $i {=} k_1$;
\item
$\ap_{\phi^2} \in \lab'((s,i))$ if and only if $s \in \Sat(\phi^2)$;
\end{itemize}
\item
for any $(s,i) \in S'$ and $a \in (A_1(s) \cup \{\bot\}){\times}(A_2(s) \cup \{\bot\})$:
\begin{itemize}
\item
$r_A^{1'}((s,i))(a)=r_A^1(s)$ if $0{\leq}i{\leq}k_1{-}1$ and equals 0 otherwise;
\item
$r_S^{1'}((s,i))=r_S^1(s)$ if $0{\leq}i{\leq}k_1{-}1$ and equals 0 otherwise;
\item
$r^{2'}_A((s,i),a)=r^2_A(s)(a)$ and $r^{2'}_S((s,i))=r_S^2(s)$.
\end{itemize}
\end{itemize}
Then compute values of $\nashop{C{:}C'}{\sim x}{\rewop{r_1'}{}{\future \ap_{k_1}}{+}\rewop{r_2'}{}{\future \ap_{\phi^2}}}$ for the game $\game'$.
\section{Correctness of the Model Checking Algorithm}\label{appx-correctness}

We fix a game $\game$ and Nash formula $\nashop{C{:}C'}{\sim x}{\theta}$. For the case of bounded zero-sum formulae the correctness of the model checking algorithm follows from the fact that we use backwards induction~\cite{SW+01,NMK+44}. Below we consider probabilistic and expected reachability objectives. The case for probabilistic until objectives follows similarly. For any $(v_1,v_2),(v_1',v_2') \in \Rset^2$, let $(v_1,v_2)\leq(v_1',v_2')$ if and only if $v_1 \leq v_1'$ and $v_2 \leq v_2'$. The following lemma follows by definition of subgame perfect SWNE values.
\begin{lemma}\label{probrew-lem}
Consider any strategy profile $\sigma$ and state $s$ of $\game^C$ and let $(v_1^{\sigma,s},v_2^{\sigma,s})$ be the corresponding values of the players in $s$ for the objectives $(X^{\theta_1},X^{\theta_2})$. Considering subgame perfect SWNE values of the objectives $(X^{\theta_1},X^{\theta_2})$ in state $s$, in the case that $\theta$ is of the form $\probop{}{\future \phi^1}{+}\probop{}{\future \phi^2}:$
\begin{itemize}
\item
if $s \sat \phi^1 \wedge \phi^2$, then $(1,1)$ are the unique subgame perfect SWNE values and $(v_1^{\sigma,s},v_2^{\sigma,s}) \leq (1,1)$;
\item
if $s \sat \phi^1 \wedge \neg \phi^2$, then $(1,\probopP^{\max}_{\game,s}(\future \phi^2))$ are the unique subgame perfect SWNE values and $(v_1^{\sigma,s},v_2^{\sigma,s}) \leq (1,\probopP^{\max}_{\game,s}(\future \phi^2))$;
\item
if $s \sat \neg \phi^1 \wedge \phi^2$, then $(\probopP^{\max}_{\game,s}(\future \phi^1),1)$ are the unique subgame perfect SWNE values for $s$ and $(v_1^{\sigma,s},v_2^{\sigma,s}) \leq (\probopP^{\max}_{\game,s}(\future \phi^1),1)$.
\end{itemize}
On the other hand, in the case that $\theta$ is of the form $\rewop{r_1}{}{\future \phi^1}{+}\rewop{r_2}{}{\future \phi^2}:$
\begin{itemize}
\item
if $s \sat \phi^1 \wedge \phi^2$, then $(0,0)$ are the unique subgame perfect SWNE values and $(v_1^{\sigma,s},v_2^{\sigma,s}) \leq (0,0)$;
\item
if $s \sat \phi^1 \wedge \neg \phi^2$, then $(0,\rewopR^{\max}_{\game,s}(r_2,\future \phi^2))$ are the unique subgame perfect SWNE values and $(v_1^{\sigma,s},v_2^{\sigma,s}) \leq (0,\rewopR^{\max}_{\game,s}(r_2,\future \phi^2))$;
\item
if $s \sat \neg \phi^1 \wedge \phi^2$, then $(\rewopR^{\max}_{\game,s}(r_1,\future \phi^1),0)$ are the unique subgame perfect SWNE values and $(v_1^{\sigma,s},v_2^{\sigma,s}) \leq (\rewopR^{\max}_{\game,s}(r_1,\future \phi^1),0)$.
\end{itemize}
\end{lemma}
Next we require the following objectives of $\game^C$.
\begin{definition}\label{bounded-objective-def}
For any probabilistic or expected reachability non-zero sum formula $\theta$, $1{\leq}i{\leq}2$ and $n \in \Nset$, let $X^\theta_{i,n}$ be the objective where for any path $\pi$ of $\game^C:$
\begin{eqnarray*}
X^{\probop{}{\future \! \phi^1}{+}\probop{}{\future \! \phi^2}}_{i,n}(\pi) & = & \left\{ \begin{array}{cl}
1 & \; \mbox{if $\exists k {\leq} n. \, \pi(k) \sat \phi^i$} \\
0 & \; \mbox{otherwise}
\end{array} \right. \\
X^{\rewop{r_{\scale{.75}{1}}}{}{\future \! \phi^1}{+}\rewop{r_{\scale{.75}{2}}}{}{\future \! \phi^2}}_{i,n}(\pi) & =  & \left\{ \begin{array}{cl}
\infty
& \mbox{if} \; \forall k \in \Nset . \, \pi(k) \notsat \phi \\
\mbox{$\sum_{k=0}^{k_\phi}$} \big( r_A(\pi(k),\pi[k])+r_S(\pi(k)) \big) & \mbox{if $k_\phi \leq n{-}1$} \\
0 & \mbox{otherwise}
\end{array} \right.
\end{eqnarray*}
and $k_\phi = \min \{ k{-}1 \mid k \in \Nset \wedge \pi(k) \sat \phi \}$.
\end{definition}
The following lemma demonstrates that, for a fixed strategy profile and state, the values of these objectives are non-decreasing and converge uniformly to the values of $\theta$.
\begin{lemma}\label{epsilon-lem}
For any probabilistic or expected reachability non-zero sum formula $\theta$ and $\varepsilon{>}0$, there exists $N \in \Nset$ such that for any $n {\geq} N$, $s \in S$, $\sigma \in \Sigma^1_{\game^C} {\times} \Sigma^2_{\game^C}$ and $1{\leq}i{\leq}2:$
\[
0 \ \leq \ \Eset^{\sigma}_{\game^C,s}(X^\theta_i) - \Eset^{\sigma}_{\game^C,s}(X^\theta_{i,n}) \ \leq \ \varepsilon \, .
\]
\end{lemma}
\begin{proof}
The result follows from \defdefref{sem-def}{bounded-objective-def} and \assumref{game-assum}. \qed
\end{proof}
In the proof of correctness we will use the fact that $n$ iterations of value iteration is equivalent to performing backwards induction on the following game trees.
\begin{definition}\label{trees-def}
For any state $s$ and $n \in \Nset$, let $\game^C_{n,s}$ be the game tree corresponding to playing $\game^C$ for $n$ steps when starting from state $s$ and then terminating. 
\end{definition}
We can map any strategy profile $\sigma$ of $\game^C$ to a strategy profile of $\game^C_{n,s}$ by only considering the choices of the profile over the first $n$ steps when starting from state $s$. This mapping is clearly surjective, i.e.\ we can generate all profiles of $\game^C_{n,s}$, but is not injective.
We also need the following objectives corresponding to the values computed during value iteration for the game trees of \defref{trees-def}.
\begin{definition}\label{value-objective-def}
For any probabilistic or expected reachability non-zero sum formula $\theta$, $s \in S$, $n \in \Nset$, $1{\leq}i{\leq}2$, let $Y^\theta_i$ be the objective where for any path $\pi$ of $\game^C_{n,s}:$
\begin{align*}
\lefteqn{\hspace{-0.2cm} Y^{\probop{}{\future \! \phi^1}{+}\probop{}{\future \! \phi^2}}_i(\pi) =} \\
 & \left\{ \begin{array}{cl}
1 & \; \mbox{if $\exists m {\leq} n . \, (\pi(m) \sat \phi^i \wedge \forall k {<} m . \, \pi(k) \notsat \phi^i \wedge \phi^j)$} \\
\probopP^{\max}_{\game,\pi(m)}(\future \phi^i) & \; \mbox{else if $\exists m {\leq} n . \, (\pi(m) \sat \phi^j \wedge \forall k {<} m . \, \pi(k) \notsat \phi^i \wedge \phi^j)$}\\
0 & \; \mbox{otherwise} 
\end{array} \right.
\\
\lefteqn{\hspace{-0.2cm} Y^{\rewop{r_{\scale{.75}{1}}}{}{\future \! \phi^1}{+}\rewop{r_{\scale{.75}{2}}}{}{\future \! \phi^2}}_i(\pi) =} \\
& \left\{ \begin{array}{cl}
\infty
& \; \mbox{if} \; \forall k {\leq} n . \, \pi(k) \notsat \phi \\
\mbox{$\sum_{k=0}^{k_\phi}$} \big( r_A(\pi(k),\pi[k]) + r_S(\pi(k)) \big) + r'_S(\pi(k)) & \; \mbox{otherwise}
\end{array} \right.
\end{align*}
where
\[
r'_S(s) \; = \; \left\{ \begin{array}{cl}
\rewopR^{\max}_{\game,\pi(m)}(r_i,\future \phi^{i}) & \; \mbox{if $s \sat \neg \phi^i \wedge \phi^j$}\\
0 & \; \mbox{otherwise} 
\end{array} \right.
\]
$k_\phi = \min \{ k{-}1 \mid k {\leq} n \wedge \pi(k) \sat \phi^1 \vee \phi^2 \}$ and $j= i{+}1 \bmod 2$.
\end{definition}
Similarly to \lemref{epsilon-lem}, the lemma below demonstrates for a fixed strategy profile and state $s$ of $\game^C$, the values for the objectives given in \defref{value-objective-def} when played on the game trees $\game^C_{n,s}$ are non-decreasing and converge uniformly. As with \lemref{epsilon-lem} the result follows from \assumref{game-assum}.
\begin{lemma}\label{epsilon2-lem}
For any probabilistic or expected reachability non-zero sum formula $\theta$ and $\varepsilon{>}0$, there exists $N \in \Nset$ such that for any $m {\geq} n {\geq} N$, $\sigma \in \Sigma^1_{\game^C} {\times} \Sigma^2_{\game^C}$, $s \in S$ and $1{\leq}i{\leq}2:$
\[
0 \ \leq \ \Eset^{\sigma}_{\game^C_{m,s}}(Y^\theta_i) - \Eset^{\sigma}_{\game^C_{n,s}}(Y^\theta_i) \ \leq \ \varepsilon  \, .
\]
\end{lemma}
We require the following lemma relating the values of the objectives $X^\theta_{i,n}$, $Y^\theta_i$ and $X^\theta_i$.
\begin{lemma}\label{precomp-lem}
For any probabilistic or expected reachability non-zero sum formula $\theta$, state $s$ of $\game^C$, strategy profile $\sigma$ such that when one of the targets of the objectives of $\theta$ is reached, the profile then collaborates to maximises the value of the other objective, $n \in \Nset$ and $1{\leq}i{\leq}2:$
\[
\sup\nolimits_{\sigma_i \in \Sigma_i^{\game^{\scale{.75}{C}}_{\scale{.75}{n,s}}}} \Eset^{\sigma_{-i}[\sigma_i]}_{\game^C,s}(X^\theta_{i,n}) \ \leq \ \sup\nolimits_{\sigma_i \in \Sigma_i^{\game^{\scale{.75}{C}}_{\scale{.75}{n,s}}}} \Eset^{\sigma_{-i}[\sigma_i]}_{\game^C_{n,s}}(Y^\theta_i)
\ \leq \
\sup\nolimits_{\sigma_i \in \Sigma_i^{\game^{\scale{.75}{C}}}} \Eset^{\sigma_{-i}[\sigma_i]}_{\game^C,s}(X^\theta_i) \, .
\]
\end{lemma}
\begin{proof}
The proof follows from the restriction on the strategy profiles and Definitions~\ref{sem-def}, \ref{bounded-objective-def} and \ref{value-objective-def}. \qed
\end{proof}
We now define the strategy profiles synthesised during value iteration.
\begin{definition}
For any $n \in \Nset$ and $s \in S$, let $\sigma^{n,s}$ be the strategy profile generated for the game tree $\game^C_{n,s}$ (when considering value iteration as backwards induction) and $\sigma^{n,\star}$ be the synthesised strategy profile for $\game^C$ after $n$ iterations.
\end{definition}
Before giving the proof of correctness we require the following results.
\begin{lemma}\label{backwards-lem}
For any state $s$ of $\game^C$, probabilistic or expected reachability non-zero sum formula $\theta$ and $n \in \Nset$ we have that $\sigma^{n,s}$ is a subgame perfect SWNE of the CSG $\game^C_{n,s}$ for the objectives $(Y^{\theta_1},Y^{\theta_2})$.
\end{lemma}
\begin{proof}
The result follows from the fact that the value iteration selects SWNEs, value iteration corresponds to performing backwards induction for the objectives $(Y^{\theta_1},Y^{\theta_2})$ and backwards induction returns a subgame perfect NE~\cite{SW+01,NMK+44}. \qed
\end{proof}
\begin{lemma}\label{strats-lem}
For any $\varepsilon{>}0$, there exists $N \in \Nset$ such that for any $s \in S$ and $1{\leq}i{\leq}2$:
\[
\big| \, \Eset^{\sigma^{n,\star}}_{\game^C,s}(X^\theta_i) - \Eset^{\sigma^{n,s}}_{\game^C_{n,s}}(Y^\theta_i) \, \big|
\ \leq \ \varepsilon \, .
\]
\end{lemma}
\begin{proof}
Using \lemref{epsilon2-lem} and since value iteration converges (\propref{convergence-prop}), we can choose $N$ such that the choices of the profile $\sigma^{n,s}$ agree with those of $\sigma^{n,\star}$ for a sufficient number of steps such that the inequality holds. \qed
\end{proof}
The following proposition demonstrates value iteration convergences and depends on \assumref{game-assum}. Without this assumption convergence cannot be guaranteed as demonstrated by the counterexamples in \appref{2-app} and \appref{3-app}. Although value iteration converges, unlike value iteration for MDPs or zero-sum games, the generated sequence of values is not necessarily non-decreasing. 
\begin{proposition}\label{convergence-prop}
For any probabilistic or expected reachability non-zero sum formula $\theta$ and state $s$, the sequence $\langle \V_{\game^C}(s,\theta,n) \rangle_{n \in \Nset}$ converges.
\end{proposition}
\begin{proof}
For any state $s$ and $n \in \Nset$ we can consider $\game^C_{n,s}$ as two-player infinite-action NFG $\nfgame_{n,s}$ where for $1 {\leq} i {\leq} 2$:
\begin{itemize}
\item
the set of actions of player $i$ equals the set of strategies of player $i$ in $\game^C$;
\item for the action pair $(\sigma_1,\sigma_2)$, the utility function for player $i$ returns $\Eset^{\sigma}_{\game^C_{n,s}}(Y^\theta_i)$.
\end{itemize}
The correctness of this construction rely on the mapping of strategy profiles from $\game^C$ to $\game^C_{n,s}$ being surjective.  Using \lemref{epsilon2-lem}, we have that sequence $\langle \nfgame_{n,s} \rangle_{n \in N}$ of NFGs converge uniformly, and therefore, since $\V_{\game^C}(s,\theta,n)$ are subgame perfect SWNE values of $\game^C_{n,s}$ (see \lemref{backwards-lem}), the sequence $\langle \V_{\game^C}(s,\theta,n) \rangle_{n \in \Nset}$ also converges. \qed
\end{proof}
A similar convergence result to \propref{convergence-prop} has been shown for discounted properties in~\cite{FL83}.
\begin{theorem}
For a given probabilistic or expected reachability non-zero sum formula $\theta$ and $\varepsilon{>}0$, there exists $N \in \Nset$ such that for any $n {\geq} N$ the strategy profile $\sigma^{n,\star}$ is a subgame perfect $\varepsilon$-SWNE for $\game^C$ and the objectives $(X^{\theta_1},X^{\theta_2})$.
\end{theorem}
\begin{proof}
Consider any $\varepsilon{>}0$. From \lemref{strats-lem} there exists $N_1 \in \Nset$ such that for any $s\in S$ and $n {\geq} N_1$:
\begin{equation}\label{1-eqn}
\big| \, \Eset^{\sigma^{n,\star}}_{\game^C,s}(X^\theta_i) - \Eset^{\sigma^{n,s}}_{\game^C_{n,s}}(Y^\theta_i) \, \big|
\ \leq \ \mbox{$\frac{\varepsilon}{2}$} \, .
\end{equation}
For any $m \in \Nset$ and $s \in S$, using \lemref{backwards-lem} we have that $\sigma^{m,s}$ is a NE of $\game^C_{m,s}$, and therefore for any $m \in \Nset$, $s\in S$ and $1{\leq}i{\leq}2$:
\begin{equation}\label{nash-eqn}
\Eset^{\sigma^{m,s}}_{\game^{C}_{m,s}}(Y^\theta_i) 
\ \geq \
\sup\nolimits_{\sigma_i \in \Sigma_i^{\game^{\scale{.75}{C}}_{\scale{.75}{m,s}}}} \Eset^{\sigma^{m,s}_{-i}[\sigma_i]}_{\game^C_{m,s}}(Y^\theta_i) \, .
\end{equation}
From \lemref{epsilon-lem} there exists $N_2 \in \Nset$ such that for any $n {\geq} N_2$, $s \in S$ and $1{\leq}i{\leq}2$:
\begin{equation}\label{2-eqn}
\sup\nolimits_{\sigma_i \in \Sigma_i^{\game^{\scale{.75}{C}}}} \Eset^{\sigma^{n,\star}_{-i}[\sigma_i]}_{\game^C,s}(X^\theta_i) - \sup\nolimits_{\sigma_i \in \Sigma_i^{\game^{\scale{.75}{C}}}} \Eset^{\sigma^{n,\star}_{-i}[\sigma_i]}_{\game^C,s}(X^\theta_{i,n}) 
\ \leq \ \mbox{$\frac{\varepsilon}{2}$} \, .
\end{equation}
By construction $\sigma^{n,\star}$ is a profile for which if one of the targets of the objectives of $\theta$ is reached the profile maximises the value of the objective, we can rearrange \eqnref{2-eqn} and apply \lemref{precomp-lem} to yield for any $n {\geq} N_2$, $s \in S$ and $1{\leq}i{\leq}2$:
\begin{equation}\label{3-eqn}
\sup\nolimits_{\sigma_i \in \Sigma_i^{\game^{\scale{.75}{C}}_{\scale{.75}{n,s}}}} \Eset^{\sigma^{n,s}_{-i}[\sigma_i]}_{\game^C_{n,s}}(Y^\theta_i)
\ \geq \
\sup\nolimits_{\sigma_i \in \Sigma_i^{\game^{\scale{.75}{C}}}} \Eset^{\sigma^{n,\star}_{-i}[\sigma_i]}_{\game^C,s}(X^\theta_i) - \mbox{$\frac{\varepsilon}{2}$} \, .
\end{equation}
Letting $N = \max \{ N_1 , N_2 \}$, for any $n {\geq} N$, $s \in S$ and $1{\leq}i{\leq}2$:
\begin{align*}
\Eset^{\sigma^{n,\star}}_{\game^C,s}(X^\theta_i) \ \ & \geq \ \ \Eset^{\sigma^{n,s}}_{\game^C_{n,s}}(Y^\theta_i) - \mbox{$\frac{\varepsilon}{2}$} & \mbox{by \eqnref{1-eqn} since $n {\geq}N_1$} \\
&\geq \ \ \sup\nolimits_{\sigma_i \in \Sigma_i^{\game^{\scale{.75}{C}}_{\scale{.75}{n,s}}}} \Eset^{\sigma^{n,s}_{-i}[\sigma_i]}_{\game^C_{n,s}}(Y^\theta_i) - \mbox{$\frac{\varepsilon}{2}$} & \mbox{by \eqnref{nash-eqn}} \\
&\geq \ \ \left( \sup\nolimits_{\sigma_i \in \Sigma_i^{\game^{\scale{.75}{C}}}} \Eset^{\sigma^{n,\star}_{-i}[\sigma_i]}_{\game^C,s}(X^\theta_i) - \mbox{$\frac{\varepsilon}{2}$} \right) - \mbox{$\frac{\varepsilon}{2}$} & \mbox{by \eqnref{3-eqn} since $n {\geq} N_2$} \\
& = \ \ \sup\nolimits_{\sigma_i \in \Sigma_i^{\game^{\scale{.75}{C}}}} \Eset^{\sigma^{n,\star}_{-i}[\sigma_i]}_{\game^C,s}(X^\theta_i) - \varepsilon
\end{align*}
and hence, since $\varepsilon{>}0$, $s \in S$ and $1{\leq}i{\leq}2$ were arbitrary, $\sigma^{n,\star}$ is a subgame perfect $\varepsilon$-NE. It remains to show the strategy profile is subgame perfect social welfare optimal $\varepsilon$-NE which follows from the fact that when solving the bimatrix games during value iteration social welfare optimal NE are returned. \qed
\end{proof}

\section{Convergence of Probabilistic Reachability Properties}\label{2-app}

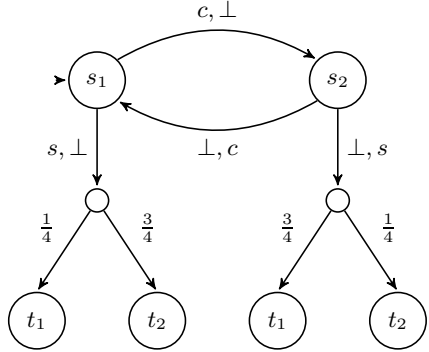
\begin{figure}[t]
\centering
\begin{tikzpicture}[->,>=stealth',shorten >=1pt,auto,node distance=2.8cm, semithick, scale=.40]

  \tikzstyle{every state}=[draw=black,text=black, initial text=]
\small 
\node[initial,state,minimum width=0.75cm,minimum height=0.75cm] (S)at(0,0) (s0) {$s_1$}; 

\node[state,minimum width=0.75cm,minimum height=0.75cm] (S)at(8,0) (s1) {$s_2$}; 

\node[state,minimum width=0cm,minimum height=0cm] (S)at(0,-4) (d1) {}; 

\node[state,minimum width=0cm,minimum height=0cm] (S)at(8,-4) (d2) {}; 

\node[state,minimum width=0.75cm,minimum height=0.75cm] (S)at(-2,-8) (t11) {$t_1$}; 

\node[state,minimum width=0.75cm,minimum height=0.75cm] (S)at(2,-8) (t12) {$t_2$}; 

\node[state,minimum width=0.75cm,minimum height=0.75cm] (S)at(6,-8) (t21) {$t_1$}; 

\node[state,minimum width=0.75cm,minimum height=0.75cm] (S)at(10,-8) (t22) {$t_2$};

\path [->] (s0.north east) [bend left]
edge node []  {$c,\bot$} (s1.north west);

\path [->] (s1.south west) [bend left]
edge node []  {$\bot,c$} (s0.south east);

\path [->] (s0.south) []
edge node [swap]  {$s,\bot$} (d1);

\path [->] (s1.south) []
edge node []  {$\bot,s$} (d2);

\path [->] (d1) []
edge node [swap]  {$\frac{1}{4}$} (t11.north);

\path [->] (d1) []
edge node []  {$\frac{3}{4}$} (t12.north);

\path [->] (d2) []
edge node [swap]  {$\frac{3}{4}$} (t21.north);

\path [->] (d2) []
edge node []  {$\frac{1}{4}$} (t22.north);

\end{tikzpicture} 
\caption{Counterexample for probabilistic reachability non-zero sum formula}\label{counter3-fig}
\end{figure}

We now present an example demonstrating that without the requirement that there are no non-terminal end components, i.e.\ when \assumref{game-assum} does not hold, value iteration can fail to converge for probabilistic reachability properties.

Consider the CSG in \figref{counter3-fig} (an adapted of a TSG example from~\cite{BMS16}) and the non-zero sum formula $\theta=\probop{}{\future \ap_1}{+}\probop{}{\future \ap_2}$ where $\ap_i$ is the atomic proposition satisfied only by the state $t_i$. Clearly this CSG has a non-terminal end component as one can remain in $\{ s_1 , s_2 \}$ indefinitely or leave at any time.

Applying the value iteration algorithm of \sectref{mc-sect} we have:
\begin{itemize}
\item In the first iteration $\V_{\game^C}(s_1,\theta,1)$ are the SWNE values of the bimatrix game:
\[
\mgame_1 =\bordermatrix{ ~ & \bot \cr
c & 0 \cr
s & \frac{1}{4}  \cr
} 
\quad \mbox{and} \quad
\mgame_2 =\bordermatrix{ ~ & \bot \cr
c & 0 \cr
s & \frac{3}{4}  \cr
}
\]
i.e. the values $(\frac{1}{4},\frac{3}{4})$, and $\V_{\game^C}(s_2,\theta,1)$ are the SWNE values of the bimatrix game:
\[
\mgame_1 =\bordermatrix{ ~ & c & s \cr
\bot & 0 & \frac{3}{4} \cr
}
\quad \mbox{and} \quad
\mgame_2 =\bordermatrix{ ~ & c & s \cr
\bot & 0 & \frac{1}{4} \cr
}
\]
i.e. the values $\
{(\frac{3}{4},\frac{1}{4})}$.
\item In the second iteration $\V_{\game^C}(s_1,\theta,2)$ are the SWNE values of the bimatrix game:
\[
\mgame_1 =\bordermatrix{ ~ & \bot \cr
c & \frac{3}{4}  \cr
s & \frac{1}{4}  \cr
}
\quad \mbox{and} \quad
\mgame_2 =\bordermatrix{ ~ & \bot \cr
c & \frac{1}{4} \cr
s & \frac{3}{4}  \cr
}
\]
i.e. the values $(\frac{3}{4},\frac{1}{4})$, and $\V_{\game^C}(s_2,\theta,2)$ are the SWNE values of the bimatrix games:
\[
\mgame_1 =\bordermatrix{ ~ & c & s \cr
\bot & \frac{1}{4} & \frac{3}{4} \cr
}
\quad \mbox{and} \quad
\mgame_2 = \bordermatrix{ ~ & c & s \cr
\bot & \frac{3}{4} & \frac{1}{4} \cr
}
\]
i.e. the values $(\frac{1}{4},\frac{3}{4})$.
\item 
In the third iteration $\V_{\game^C}(s_1,\theta,3)$ are the SWNE values of the bimatrix game:
\[
\mgame_1 =\bordermatrix{ ~ & \bot \cr
c & \frac{1}{4}  \cr
s & \frac{1}{4}  \cr
}
\quad \mbox{and} \quad
\mgame_2 =\bordermatrix{ ~ & \bot \cr
c & \frac{3}{4} \cr
s & \frac{3}{4}  \cr
}
\]
i.e. the values $(\frac{1}{4},\frac{3}{4})$, and $\V_{\game^C}(s_2,\theta,3)$ are the SWNE values of the bimatrix game:
\[
\mgame_1 =\bordermatrix{ ~ & c & s \cr
\bot & \frac{3}{4} & \frac{3}{4} \cr
}
\quad \mbox{and} \quad
\mgame_2 = \bordermatrix{ ~ & c & s \cr
\bot & \frac{1}{4} & \frac{1}{4} \cr
}
\]
i.e. the values $(\frac{3}{4},\frac{1}{4})$. 
\item
In the fourth iteration $\V_{\game^C}(s_1,\theta,4)$ are the SWNE values of the bimatrix game:
\[
\mgame_1 =\bordermatrix{ ~ & \bot \cr
c & \frac{3}{4}  \cr
s & \frac{1}{4}  \cr
}
\quad \mbox{and} \quad
\mgame_2 =\bordermatrix{ ~ & \bot \cr
c & \frac{1}{4} \cr
s & \frac{3}{4}  \cr
}
\]
i.e. the values $(\frac{3}{4},\frac{1}{4})$, and $\V_{\game^C}(s_2,\theta,4)$ are the SWNE values of the bimatrix game:
\[
\mgame_1 =\bordermatrix{ ~ & c & s \cr
\bot & \frac{3}{4}  & \frac{3}{4} \cr
}
\quad \mbox{and} \quad
\mgame_2 = \bordermatrix{ ~ & c & s \cr
\bot & \frac{1}{4} & \frac{1}{4} \cr
}
\]
i.e. the values $(\frac{3}{4},\frac{1}{4})$.
\end{itemize}
As can be seen the values computed at each iteration for the states $s_1$ and $s_2$ will oscillate between $(\frac{1}{4},\frac{3}{4})$ and $(\frac{3}{4},\frac{1}{4})$.
\section{Convergence of Expected Reachability Properties}\label{3-app}

\begin{figure}[t]
\begin{subfigure}{.4\textwidth}
\centering
\begin{tikzpicture}[->,>=stealth',shorten >=1pt,auto,node distance=2.8cm, semithick, scale=.40]
  \tikzstyle{every state}=[draw=black,text=black, initial text=]
\small 
\node[initial,state,minimum width=0.75cm,minimum height=0.75cm] (S)at(0,0) (s0) {$s_1$}; 

\node[state,minimum width=0.75cm,minimum height=0.75cm] (S)at(8,0) (s1) {$s_2$}; 

\node[state,minimum width=0.75cm,minimum height=0.75cm] (S)at(0,-4) (s2) {$t_1$}; 

\node[state,minimum width=0.75cm,minimum height=0.75cm] (S)at(8,-4) (s3) {$t_2$}; 

\path [->] (s0.north east) [bend left]
edge node []  {$c,\bot$} (s1.north west);

\path [->] (s1.south west) [bend left]
edge node []  {$\bot,c$} (s0.south east);

\path [->] (s0.south) []
edge node [swap]  {$s,\bot$} (s2.north);

\path [->] (s1.south) []
edge node []  {$\bot,s$} (s3.north);

\end{tikzpicture} 
\end{subfigure}
\begin{subfigure}{.6\textwidth}
\begin{align*}
\\
r_A^1(s, (a_1,a_2) ) & = \; \left\{ \begin{array}{cl}
\frac{1}{3} & \mbox{if $s=s_1$ and $(a_1,a_2)=(s,\bot)$} \\
2 & \mbox{if $s=s_2$ and $(a_1,a_2)=(\bot,s)$} \\
0 & \mbox{otherwise}
\end{array} \right. \\
r_A^2(s, (a_1,a_2) ) & = \; \left\{ \begin{array}{cl}
1 & \mbox{if $s=s_1$ and $(a_1,a_2)=(s,\bot)$} \\
\frac{1}{3} & \mbox{if $s=s_2$ and $(a_1,a_2)=(\bot,s)$} \\
0 & \mbox{otherwise}
\end{array} \right. \\
\end{align*}
\end{subfigure}
\caption{Counterexample for expected reachability non-zero sum formula}\label{counter1-fig}
\end{figure}
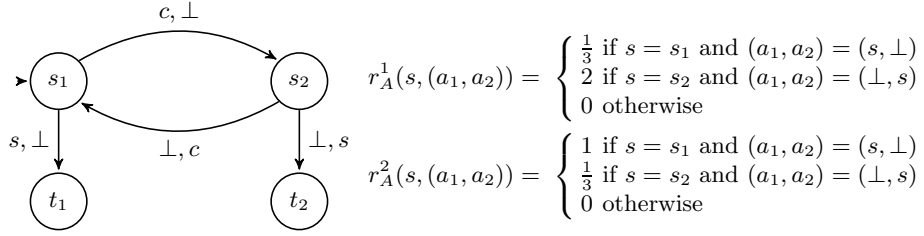

In this section we present an example demonstrating that without the requirement that the targets are reached with probability 1 under all strategy profiles, i.e.\ when \assumref{game-assum} does not hold, value iteration can fail to converge for expected reachability properties.

Consider the CSG in \figref{counter1-fig} (which again is an adaption of a TSG example from~\cite{BMS16}) and the formula $\theta=\rewop{r_1}{}{\future \ap}{+}\rewop{r_2}{}{\future \ap}$ where $\ap$ is the atomic proposition satisfied only by the states $t_1$ and $t_2$. Clearly there are strategy profiles for which the targets are not reached with probability 1. 

Applying the value iteration algorithm of \sectref{mc-sect} we have:
\begin{itemize}
\item In the first iteration $\V_{\game^C}(s_1,\theta,1)$ are the SWNE values of the bimatrix game:
\[
\mgame_1 =\bordermatrix{ ~ & \bot \cr
c & 0 \cr
s & \frac{1}{3}  \cr
} 
\quad \mbox{and} \quad
\mgame_2 =\bordermatrix{ ~ & \bot \cr
c & 0 \cr
s & 1  \cr
}
\]
i.e. the values $(\frac{1}{3},1)$, and $\V_{\game^C}(s_2,\theta,1)$ are the SWNE values of the bimatrix game:
\[
\mgame_1 =\bordermatrix{ ~ & c & s \cr
\bot & 0 & 2 \cr
}
\quad \mbox{and} \quad
\mgame_2 =\bordermatrix{ ~ & c & s \cr
\bot & 0 & \frac{1}{3} \cr
}
\]
i.e. the values $\
{(2,\frac{1}{3})}$.
\item In the second iteration $\V_{\game^C}(s_1,\theta,2)$ are the SWNE values of the bimatrix game:
\[
\mgame_1 =\bordermatrix{ ~ & \bot \cr
c & 2 \cr
s & \frac{1}{3}  \cr
}
\quad \mbox{and} \quad
\mgame_2 =\bordermatrix{ ~ & \bot \cr
c & \frac{1}{3} \cr
s & 1  \cr
}
\]
i.e. the values $(2,\frac{1}{3})$, and $\V_{\game^C}(s_2,\theta,2)$ are the SWNE values of the bimatrix games:
\[
\mgame_1 =\bordermatrix{ ~ & c & s \cr
\bot & \frac{1}{3} & 2 \cr
}
\quad \mbox{and} \quad
\mgame_2 = \bordermatrix{ ~ & c & s \cr
\bot & 1 & \frac{1}{3} \cr
}
\]
i.e. the values $(\frac{1}{3},1)$.
\item 
In the third iteration $\V_{\game^C}(s_1,\theta,3)$ are the SWNE values of the bimatrix game:
\[
\mgame_1 =\bordermatrix{ ~ & \bot \cr
c & \frac{1}{3}  \cr
s & \frac{1}{3}  \cr
}
\quad \mbox{and} \quad
\mgame_2 =\bordermatrix{ ~ & \bot \cr
c & 1 \cr
s & 1  \cr
}
\]
i.e. the values $(\frac{1}{3},1)$, and $\V_{\game^C}(s_2,\theta,3)$ are the SWNE values of the bimatrix game:
\[
\mgame_1 =\bordermatrix{ ~ & c & s \cr
\bot & 2 & 2 \cr
}
\quad \mbox{and} \quad
\mgame_2 = \bordermatrix{ ~ & c & s \cr
\bot & \frac{1}{3} & \frac{1}{3} \cr
}
\]
i.e. the values $(2,\frac{1}{3})$. 
\item
In the fourth iteration $\V_{\game^C}(s_1,\theta,4)$ are the SWNE values of the bimatrix game:
\[
\mgame_1 =\bordermatrix{ ~ & \bot \cr
c & 2  \cr
s & \frac{1}{3}  \cr
}
\quad \mbox{and} \quad
\mgame_2 =\bordermatrix{ ~ & \bot \cr
c & \frac{1}{3} \cr
s & 1  \cr
}
\]
i.e. the values $(2,\frac{1}{3})$, and $\V_{\game^C}(s_2,\theta,4)$ are the SWNE values of the bimatrix game:
\[
\mgame_1 =\bordermatrix{ ~ & c & s \cr
\bot & \frac{1}{3}  & 2 \cr
}
\quad \mbox{and} \quad
\mgame_2 = \bordermatrix{ ~ & c & s \cr
\bot & 1 & \frac{1}{3} \cr
}
\]
i.e. the values $(\frac{1}{3},1)$.
\end{itemize}
As can be seen the values computed during value iteration  are oscillating both for state $s_1$ and $s_2$.

}

\end{document}